\documentclass[11pt,a4paper]{article}

\usepackage[utf8]{inputenc}
\usepackage[T1]{fontenc}

\usepackage[english]{babel}
\usepackage[autostyle]{csquotes}     

\usepackage[
  a4paper,
  top=2cm, bottom=2cm,
  left=3cm, right=3cm,
  marginparwidth=1.75cm
]{geometry}

\usepackage{xcolor}
\usepackage{lipsum}      
\usepackage{setspace}
\usepackage{lineno}

\usepackage{amsmath, amssymb, mathtools, amsthm, euscript}

\usepackage{graphicx}
\usepackage[section]{placeins}  
\usepackage{float}             

\usepackage{hyperref}
\usepackage[capitalise,nameinlink]{cleveref} 

\usepackage{xifthen}
\usepackage{needspace}
\usepackage{tikz}
\usepackage{authblk}

\usepackage[
  backend=biber,
  style=alphabetic,
  maxbibnames=99,
  maxcitenames=3,
  giveninits=false
]{biblatex}
\renewbibmacro*{name:andothers}{}  
\newtheorem{Theorem}{Theorem}
\newtheorem{Lemma}{Lemma}

\newtheorem{Claim}{Claim}

\newtheorem{Definition}{Definition}

\newtheorem{Corollary}{Corollary}
\newtheorem{Observation}{Observation}

\newcommand{\rad}{\text{rad}}
\newcommand{\OPT}{\textsc{OPT}}

\begin{document}
\title{Improved Lower Bounds on Multiflow-Multicut Gaps \footnote{A preliminary version of this paper appeared in the proceedings of APPROX-RANDOM 2025~\cite{kalantarzadeh_et_al:LIPIcs.APPROX/RANDOM.2025.14}}}

\author{
  Sina Kalantarzadeh \\
  University of Waterloo, Canada \\
  \texttt{s4kalant@uwaterloo.ca}
  \and
  Nikhil Kumar \\
  Tata Institute of Fundamental Research, India \\
  \texttt{kumar.nikhil@tifr.res.in}
}

\maketitle
\begin{abstract}
\noindent Given a set of source-sink pairs, the maximum \emph{multiflow} problem asks for the maximum total amount of flow that can be feasibly routed between them. The minimum \emph{multicut}, a dual problem to multiflow, seeks the minimum-cost set of edges whose removal disconnects all the source-sink pairs. It is easy to see that the value of the minimum multicut is at least that of the maximum multiflow, and their ratio is called the \emph{multiflow-multicut gap}. The classical max-flow min-cut theorem states that when there is only one source-sink pair, the gap is exactly one. However, in general, it is well known that this gap can be arbitrarily large. In this paper, we study this gap for classes of planar graphs and establish improved lower bound results. In particular, we show that this gap is at least $\frac{16}{7}$ for the class of planar graphs, improving upon the decades-old lower bound of 2. More importantly, we develop new techniques for proving such a lower bound, which may be useful in other settings as well.
\end{abstract}

\newpage

\section{Introduction}

Given an edge-weighted graph with \( k \) source-sink pairs, a \emph{multicut} is a set of edges whose removal disconnects all the source-sink pairs. The \emph{minimum multicut problem} seeks a multicut with the minimum total edge weight. This problem generalizes the classical minimum \( s \)-\( t \) cut problem and has been extensively studied in the past. Computing the minimum multicut is NP-hard, even in highly restricted settings such as trees~\cite{garg1997primal}.

\noindent A problem closely related to the multicut problem is the multicommodity flow problem (also known as \emph{multiflow}). The goal of this problem is to maximize the total flow that can be routed between the source-sink pairs. If the flow is restricted to take only integer values, the problem is called the \emph{maximum integer multiflow problem}, which generalizes the well-known \emph{edge-disjoint paths} problem.  
Since any source-sink path must use at least one edge of any multicut, the value of any feasible multicut is at least that of the maximum multicommodity flow.
In fact, it turns out that the LP relaxation of multicut problem is the linear programming dual of the multiflow problem.  The ratio of the minimum multicut to the maximum multicommodity flow is called the \emph{multiflow-multicut gap}. By the strong duality of linear programming, it follows that the integrality gap of the natural linear programming relaxation for the multicut also provides a bound on the multiflow-multicut gap, and vice versa.  

The famous max-flow min-cut theorem \cite{ford1956maximal} states that the multiflow-multicut gap is exactly 1 when \( k = 1 \), i.e., when there is exactly one source-sink pair. A well-known theorem by Hu \cite{hu1963multi} further establishes that the gap remains 1 when \( k = 2 \). However, this equality does not hold when there are three or more source-sink pairs, even for very simple graphs (see \cite{garg1997primal} for an example).  

\noindent Garg, Vazirani, and Yannakakis~\cite{garg1996approximate} proved a tight bound of \(\Theta(\ln k\)) on the multiflow-multicut gap for any graph \( G \). If \( G \) is a tree, then the multiflow-multicut gap is exactly \( 2 \) \cite{garg1997primal}. For \( K_r \)-minor-free graphs, Tardos and Vazirani~\cite{tardos1993improved} used the decomposition theorem of Klein, Plotkin, and Rao~\cite{klein1993excluded} to prove a bound of \(\mathcal{O}(r^3)\) on the multiflow-multicut gap. This bound was subsequently improved to \(\mathcal{O}(r^2)\) by Fakcharoenphol and Talwar~\cite{fakcharoenphol2003improved}, and then to \(\mathcal{O}(r)\) by Abraham et al.~\cite{abraham2014cops}. A tight bound of \(\Theta(\log r)\) was then obtained for graphs of bounded treewidth~\cite{filtser2024optimal, friedrich2023approximate}. Finally, building upon this long sequence of results, Conroy and Filtser~\cite{conroy2025protect} recently proved an asymptotically tight bound of \( \Theta(\log r) \) on the multiflow--multicut gap for \( K_r \)-minor-free graphs. Since planar graphs do not contain \( K_5 \) as a minor, it follows that the integrality gap of the minimum multicut problem for planar graphs is \(\mathcal{O}(1)\).  

The primary motivation behind the works mentioned above was to establish an asymptotic bound on the integrality gap (in terms of \( r \)) without optimizing the constants involved. However, for specific graph families, such as planar graphs, the constant obtained from these results is quite large (close to 100). Thus, determining the exact integrality gap remains an intriguing question. It is known that the integrality gap is at least 2 for trees, and consequently for planar graphs as well. Better upper and lower bounds for this problem remain elusive, serving as the primary motivation for this paper.

\subsection{Related Work: Demand Multicommodity Flow} 
In another well studied version of the problem, called the \emph{demand multicommodity flow}, we are given a demand value for each source-sink pair, denoted as \( d_i \) for the source-sink pair \( s_i \)-\( t_i \). The goal is to determine whether there exists a feasible flow satisfying all the demands. A necessary condition for the existence of a feasible flow is as follows: across every bi-partition \( (S, \Bar{S}) \) of the vertex set, the total demand that must be routed across \( (S, \Bar{S}) \) must not exceed the total capacity of edges crossing \( (S, \Bar{S}) \). This condition is known as the \emph{cut-condition}, and it is a sufficient condition for the existence of flows in trees, outerplanar graphs, and similar graph classes.  

\noindent In general, however, the cut-condition is not sufficient for the existence of a feasible flow. This leads to a natural question: what is the minimum relaxation of the cut-condition that ensures feasibility? Specifically, what is the smallest \( \alpha \geq 1 \) such that if the total capacity of edges across every bi-partition is at least \( \alpha \) times the demand across the partition, then a feasible flow is guaranteed? In their seminal work, Linial, London, and Rabinovich~\cite{Linial1995} showed that this gap is \(\Theta(\log k)\) for general graphs.  

In contrast to the multiflow-multicut gap, our understanding of the flow-cut gap for planar graphs remains limited. Rao~\cite{rao1999small} showed that the flow-cut gap for planar graphs is \(\mathcal{O}(\sqrt{\log n})\). However, the best known lower bound remains just 2~\cite{lee2010coarse, chekuri2013flow}, and it is conjectured that the true answer is \(\mathcal{O}(1)\)~\cite{gupta2004cuts}.\footnote{This conjecture is widely known as the Planar Embedding Conjecture or the GNRS Conjecture~\cite{gupta2004cuts}.}  

On the other hand, we have a much better understanding of this gap for series-parallel graphs, a subclass of planar graphs. The flow-cut gap is exactly 2 for series-parallel graphs~\cite{chakrabarti2008embeddings, lee2010coarse}. Given the current state of research, one might be tempted to claim that we understand multiflow-multicut gaps better than flow-cut gaps. However, somewhat surprisingly, the precise multiflow-multicut gap for series-parallel graphs remains unknown, despite the well-understood flow-cut gap. One of the primary motivations of this paper is to bridge this gap in our understanding.  

\section{Preliminaries} \label{section: preliminaries}

Given a graph $G$, we denote its vertex and edge sets by \( V(G) \) and \( E(G) \), respectively. We will use \( K_r \) to denote the complete graph on \( r \) vertices. In this paper, we will only be concerned with planar graphs. A graph $G$ is \emph{planar} if it does not contain \( K_5 \) or \( K_{3,3} \) as a minor. Equivalently, a graph is planar if it can be drawn in the plane without any of its edges crossing. Graphs in which every edge is contained in at most one cycle are called \emph{cactus} graphs. Cactus graphs are a subclass of series-parallel and planar graphs, and are arguably the simplest family of planar graphs after trees and cycles. Cactus and series-parallel graphs do not contain \( K_4 \) as a minor.

\noindent Let $G$ be a simple undirected graph with edge costs \( c \colon E(G) \rightarrow \mathbb{Q}_{\geq 0} \), and let \( \{(s_i, t_i)\}_{i=1}^{k} \) be the set of source-sink pairs. Let \( \mathcal{P}_i \) denote the set of all paths between \( s_i \) and \( t_i \) in \( G \), and let \( \mathcal{P} = \bigcup_{i=1}^{k} \mathcal{P}_i \). A \emph{multicut} is a set of edges \( F \subseteq E(G) \) such that every \( P \in \mathcal{P} \) contains at least one edge in \( F \). Equivalently, a multicut is a set of edges whose removal disconnects every source-sink pair. A \emph{multicommodity flow} is an assignment of non-negative real numbers to the paths in \( \mathcal{P} \) that respects the capacity constraints of the edges. In the \emph{maximum multiflow problem}, the objective is to find an assignment which maximizes the total value of flow routed. 

\noindent Given two arbitrary vertices $u, v \in V(G) $, we use $d_{G}(u, v)$ to denote the shortest path distance between $u$ and $v$ in $G$, if $G$ is clear in the context then we use $d(u,v)$ for simplicity. The diameter of $G$ is the maximum distance between a pair of vertices in $G$, i.e., $ \text{diam}(G) = \max_{u, v \in V(G)} d_{G}(u, v) $. We use $d_{G}(v,e)$ to denote the distance of a vertex $v$ from an edge $ e = (x, y) $, i.e., $ d_{G}(v, e) = \min \{ d_{G}(v, x), d_{G}(v, y) \} $.

\noindent For \( F \subseteq E(G) \), we use \(G\setminus F\) to denote the remaining graph after the removal of $F$ from $G$. For any \( v \in V(G) \), we use \( C_F(v) \) to denote the connected component of \( G\setminus F \) containing \( v \). We overload notation and also use \( C_F(v) \) to denote the set of vertices in the connected component containing \( v \). We define the radius of \( v \) with respect to \( F \) as the distance of the farthest vertex from \( v \) in \( C_F(v) \), i.e., \( \text{rad}_F(v) = \max_{u \in C_F(v)} d_{G}(v, u) \). In addition, the diameter of \( F \) is the maximum diameter of a connected component after the removal of \( F \) from \( G \), i.e., \( \text{diam}(F) = \max_{v \in V(G)} \text{diam}(C_F(v)) \). Given \( t \in \mathbb{R}_{\geq 0} \) as a parameter, we say that \( F \) forms a \( t \)-diameter decomposition if \( \text{diam}(F) < t \). We denote the set of all \( t \)-diameter decompositions of \( G \) by \( \mathcal{F}_t(G) \). Note that when referring to the distance between two vertices \( u, v \) in a component \( C \), \( d_{G}(u,v) \) denotes their distance in \( G \), rather than in the subgraph induced by \( C \), i.e., \( G[C] \). 

\subsection{Linear Programming Relaxation for the Minimum Multicut Problem}
We begin by describing an integer programming (IP) formulation for the minimum multicut problem. For each edge \( e \in E(G) \), we introduce an integer variable \( x(e) \in \{0, 1\} \), which indicates whether the edge is selected in the multicut. For a given path \( P \), we define \( x(P) = \sum_{e \in E(P)} x(e) \). A feasible multicut must include at least one edge from each source-sink path, so we impose the constraint \( x(P) \geq 1 \) for all \( P \in \mathcal{P} \), ensuring that each path is \emph{cut} by at least one edge. We relax the integrality constraints to obtain the linear programming (LP) relaxation of the multicut problem, which is formulated as follows:

\begin{equation}\label{LP of multicut}
\begin{aligned}
    \min \;& \sum_{e \in E(G)} c(e)\, x(e) \\
    \text{s.t. }\;& x(P) \;\geq\; 1 && \forall\, P \in \mathcal{P}, \\
                 & x(e) \;\geq\; 0  && \forall\, e \in E(G).
\end{aligned}
\end{equation}

\noindent Even though there are an exponential number of constraints, it is well known that the optimal solution to this LP can be computed in polynomial time \cite{garg1996approximate}. We denote the optimal solutions of the integer and linear programs as \( \OPT_{IP} \) and \( \OPT_{LP} \), respectively. We refer to \( \OPT_{LP} \) as the minimum fractional multicut. We know that the value of the maximum multiflow is equal to the minimum fractional multicut. Furthermore, a bound on the integrality gap of the LP relaxation for the multicut problem provides the same bound for the multiflow-multicut gap. Therefore, from this point onward, we will focus solely on the integrality gap of the multicut LP. We now formally define the integrality gap of the minimum multicut problem on a family of graphs.

\begin{Definition}\label{Def: Integrality Gap}
Let \( \mathcal{G} \) be a family of graphs, and let \( \mathcal{M}(\mathcal{G}) \) denote the family of all instances of the minimum multicut problem on \( \mathcal{G} \), obtained by assigning arbitrary costs to the edges and selecting a set of source--sink pairs. The integrality gap \( \alpha_{\mathcal{M}(\mathcal{G})} \) of the minimum multicut problem on \( \mathcal{M}(\mathcal{G}) \) is defined as
\[
\alpha_{\mathcal{M}(\mathcal{G})} := \max_{M \in \mathcal{M}(\mathcal{G})} \frac{\OPT_{IP}(M)}{\OPT_{LP}(M)},
\]
where $\OPT_{IP}(M)$ is the optimal value of the integer program and $\OPT_{LP}(M)$ is the optimal value of its linear relaxation \eqref{LP of multicut}.
\end{Definition}

\noindent As mentioned in the introduction, \( \alpha_{\mathcal{M}(\textsc{tree})} = 2 \)~\cite{garg1997primal}, where \( \textsc{tree} \) denotes the family of all trees, and \( \alpha_{\mathcal{M}(\textsc{planar})} = \mathcal{O}(1) \)~\cite{klein1993excluded}, where \( \textsc{planar} \) denotes the family of all planar graphs.

\section{Our Results and Techniques} \label{section:result/tech}

We provide a partial answer to the questions raised above by showing that the integrality gap of the minimum multicut problem for the family of cactus graphs (and therefore for series-parallel graphs and planar graphs) is strictly greater than 2. In particular, we show that the multiflow-multicut gap is at least \( \frac{16}{7} \) for the class of \emph{cactus} graphs.

We first develop a novel technique to argue that the integrality gap of the multicut LP is at least \( \frac{20}{9} \), and later refine it to obtain an improved lower bound of $\frac{16}{7}$. We observe that the integrality gap of the multicut LP for a class of graphs is \( \alpha \) if and only if any fractional solution to the natural linear programming relaxation of the minimum multicut problem can be approximately written as a convex combination (or equivalently a probability distribution) of feasible multicuts. Furthermore, a feasible multicut can be interpreted in terms of small diameter decompositions (i.e., a set of edges whose removal results in connected components of small diameter) with an appropriate distance function. Therefore, if a graph class admits an integrality gap of at most \( \alpha \), then there exists a set of small diameter decompositions that do not cut any fixed edge too many times. We describe this in detail in Section~\ref{Sec:sdd_multicut_equiv}.

Our crucial insight is that if the integrality gap is \( \alpha \) for a class of graphs, then there exists a well-structured set of small diameter decompositions that can be used to construct the aforementioned convex combination. These structured decompositions are inspired by the well-known single-source distance-based decomposition algorithms for trees. We also describe this in detail in Section~\ref{Sec:sdd_multicut_equiv}. The final step of the proof involves using these structural insights to argue that there cannot exist a small diameter decomposition with a small value of \( \alpha \) for the family of cactus graphs. Note that this proof is non-constructive and does not lead to an explicit example with a large gap. Nevertheless, this proof provides sufficient structural insights into instances with a large integrality gap, allowing us to construct explicit examples of cactus graphs where the gap is at least \( \frac{20}{9} \) (unfortunately, we were unable to find an explicit example showing a lower bound of $\frac{16}{7}$). We emphasize that we attempted to construct these examples through an exhaustive computer search and manual crafting but were unsuccessful. Furthermore, the structural properties established in the first proof hold for very general classes of graphs, specifically those closed under edge subdivision and 1-sum operations, and may prove useful in other settings as well. 




\section{The Integrality Gap of Multicut and Small Diameter Decomposition}\label{Sec:sdd_multicut_equiv}

\noindent Let $\mathcal{G}$ be a family of graphs closed under taking minors and under edge subdivisions, and let $\mathcal{M}(\mathcal{G})$ denote the corresponding family of all minimum multicut instances on $\mathcal{G}$. Theorem~\ref{Theorem: the existence of uniform random distribution}, which is a direct application of the work of Carr and Vempala~\cite{VempalaConvexCombination} to the minimum multicut problem, shows that any feasible fractional solution to the LP relaxation can be approximately represented as a convex combination of feasible multicuts. For completeness, we include a proof, although it is not a novel contribution of this work.

\begin{Theorem}\label{Theorem: the existence of uniform random distribution}
Suppose we are given an instance $M \in \mathcal{M}(\mathcal{G})$. Let $\mathcal{F} \subseteq 2^{E(G)}$ be the set of all feasible multicuts for $M$, and let $x$ be a feasible fractional solution to the LP relaxation \eqref{LP of multicut}. Then there exists a probability distribution $y$ over $\mathcal{F}$ such that
\begin{equation*}\label{Equivalent}
    \sum_{\substack{F \in \mathcal{F} \\ e \in F}} y_{F} \;\leq\; \alpha_{\mathcal{M}(\mathcal{G})} \cdot x(e) \quad \forall e \in E(G).
\end{equation*}
\end{Theorem}

\begin{proof} 
Suppose the statement does not hold. Then the following linear system \eqref{eq: LP feasibility of random distribution} is infeasible.

\begin{equation}\label{eq: LP feasibility of random distribution}
\begin{aligned}
    \sum_{F\in \mathcal{F}} y_{F} &= 1\\ 
    \sum_{\substack{F\in \mathcal{F}\\ e\in F}} y_{F} &\leq \alpha \cdot x(e) 
        \quad\quad \forall e\in E(G)\\
    y_{F} &\geq 0
\end{aligned}
\end{equation}

\noindent This implies that the following system \eqref{eq: extended LP feasibility of random distribution} is infeasible as well. The reason is that if the system below is feasible, then we can scale down the feasible solution appropriately and obtain a feasible solution for the system above.

\begin{subequations}\label{eq: extended LP feasibility of random distribution}
\begin{align}
\sum_{F\in \mathcal{F}} y_{F} &\geq 1 
    \label{eq: sum yf>=1}\\
\sum_{\substack{F\in \mathcal{F} \\ e\in F}} y_{F} &\leq \alpha \cdot x(e) 
    \quad \forall\, e \in E(G) \nonumber\\
y_{F} &\geq 0 
    \quad \forall\, F \in \mathcal{F} \nonumber
\end{align}
\end{subequations}

\noindent By reversing the inequality \eqref{eq: sum yf>=1}, we obtain that the following system \eqref{eq: extended LP feasibility of random distribution reversing yf>=1} is also infeasible:

\begin{equation}\label{eq: extended LP feasibility of random distribution reversing yf>=1}
    \begin{aligned}
    \sum_{F\in \mathcal{F}} (-y_{F}) &\leq -1, \\
    \sum_{\substack{F\in \mathcal{F}\\ e\in F}} y_{F} &\leq \alpha \cdot x(e) \quad \forall\, e \in E(G), \\
    y_{F} &\geq 0 \quad \forall\, F \in \mathcal{F}.
\end{aligned}
\end{equation}

\noindent Now, we use the following variant of Farkas Lemma (See~\cite{SchrijverFarkasLemma} for a proof).
\begin{Lemma}\label{Farkas Lemma}
    $\{x\in \mathbb{R}^{n}|Ax\leq b,x\geq 0\}=\emptyset$ iff there exists a vector $u$ such that $A^{T} u \geq 0,u \geq 0$ and $b^{T} u <0$. 
\end{Lemma}

\noindent For a feasible multicut $F$, let $\chi_F \in \{0,1\}^E$ denote its indicator vector. By Lemma~\ref{Farkas Lemma}, there exists $u\geq 0$ and $c \geq 0$ such that $c^{T}\chi_{F}-u\geq 0$ for all $F\in \mathcal{F}$ and $-u + \alpha \cdot c^{T} x<0$. This means that $c^{T}\chi_{F}\geq u$ for all $F\in \mathcal{F}$, and $\alpha \cdot c^{T}x < u$. Thus, $c^{T}\chi_{F}\geq u>\alpha \cdot c^{T}x$ for all $F\in\mathcal{F}$. Therefore with respect to the cost function $c$, $\OPT_{LP}\leq \frac{u}{\alpha}$ and $\OPT_{IP}>u$. This implies that integrality gap of the multicut instance $M$ is $> \alpha$, a contradiction.
\end{proof}

\noindent To connect this with the notion of small diameter decompositions, we now give the formal definition.

\begin{Definition}[Small Diameter Decomposition (SDD)]
Given an unweighted graph $G$, an integer parameter $k \in \mathbb{N}$, and a probability parameter $0 < p < 1$, the small diameter decomposition (SDD) problem asks whether there exists a probability distribution $\mathcal{D} = \{y_{F}\}_{F \in \mathcal{F}_{k}(G)}$ over $\mathcal{F}_{k}(G)$, the family of $k$-diameter decompositions of $G$, such that every edge $e \in E(G)$ is included in a random $k$-diameter decomposition sampled from $\mathcal{D}$ with probability at most $p$, that is,
\[
\sum_{\substack{F \in \mathcal{F}_{k}(G) \\ e \in F}} y_{F} \;\leq\; p \quad \text{for all } e \in E(G).
\]
If such a distribution exists, we denote it by $SDD(G,k,p)$ (See Figure \ref{sdd path example label}). Moreover, a family of graphs $\mathcal{G}$ is said to be $SDD(k,p)$-acceptable if for every $G \in \mathcal{G}$ there exists an $SDD(G,k,p)$.
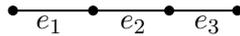
\begin{figure}[H]
    \centering
    \tikzset{every picture/.style={line width=0.75pt}} 

\begin{tikzpicture}[x=0.75pt,y=0.75pt,yscale=-1,xscale=1]

\draw  [fill={rgb, 255:red, 0; green, 0; blue, 0 }  ,fill opacity=1 ] (164,191) .. controls (164,189.9) and (164.9,189) .. (166,189) .. controls (167.1,189) and (168,189.9) .. (168,191) .. controls (168,192.1) and (167.1,193) .. (166,193) .. controls (164.9,193) and (164,192.1) .. (164,191) -- cycle ;
\draw  [fill={rgb, 255:red, 0; green, 0; blue, 0 }  ,fill opacity=1 ] (204,191) .. controls (204,189.9) and (204.9,189) .. (206,189) .. controls (207.1,189) and (208,189.9) .. (208,191) .. controls (208,192.1) and (207.1,193) .. (206,193) .. controls (204.9,193) and (204,192.1) .. (204,191) -- cycle ;
\draw  [fill={rgb, 255:red, 0; green, 0; blue, 0 }  ,fill opacity=1 ] (242,191) .. controls (242,189.9) and (242.9,189) .. (244,189) .. controls (245.1,189) and (246,189.9) .. (246,191) .. controls (246,192.1) and (245.1,193) .. (244,193) .. controls (242.9,193) and (242,192.1) .. (242,191) -- cycle ;
\draw  [fill={rgb, 255:red, 0; green, 0; blue, 0 }  ,fill opacity=1 ] (276,191) .. controls (276,189.9) and (276.9,189) .. (278,189) .. controls (279.1,189) and (280,189.9) .. (280,191) .. controls (280,192.1) and (279.1,193) .. (278,193) .. controls (276.9,193) and (276,192.1) .. (276,191) -- cycle ;
\draw    (168,191) -- (278,191) ;

\draw (176,192.4) node [anchor=north west][inner sep=0.75pt]    {$e_{1}$};
\draw (218,192.4) node [anchor=north west][inner sep=0.75pt]    {$e_{2}$};
\draw (255,192.4) node [anchor=north west][inner sep=0.75pt]    {$e_{3}$};

\end{tikzpicture}
    \caption{$G$ is a simple path with $3$ edges, and $k=2$. Let $F_{1}=\{e_{1},e_{2}\},F_{2}=\{e_{2}\},F_{3}=\{e_{2},e_{3}\},F_{4}=\{e_{1},e_{3}\}$. One can see that $\mathcal{F}_{2}(G)=\{F_{1},F_{2},F_{3},F_{4},E(G)\}$. Let $y_{F_{1}}=y_{F_{3}}=y_{E(G)}=0$, and $y_{F_{2}}=y_{F_{4}}=\frac{1}{2}$. This distribution is a $SDD(G,2,\frac{1}{2})$.}
    \label{sdd path example label}
\end{figure}
\end{Definition}

\noindent In the following Theorem~\ref{relation integrality gap and SDD}, we make explicit the relation between the integrality gap $\alpha_{\mathcal{M}(\mathcal{G})}$ and the existence of suitable SDDs for graph families $\mathcal{G}$ closed under minors and edge subdivisions. The forward implication is a direct consequence of Theorem~\ref{Theorem: the existence of uniform random distribution}, which itself follows straightforwardly from the work of Carr and Vempala~\cite{VempalaConvexCombination} in the context of the minimum multicut problem, and is therefore not our contribution. The backward implication similarly relies on the simple subdivision technique used in~\cite{tardos1993improved}, and again is not part of our new contributions. In this paper, we make use only of the forward implication of Theorem~\ref{relation integrality gap and SDD}, but we also include the proof of the backward implication for completeness.

\begin{Theorem}\label{relation integrality gap and SDD}
    Let $\mathcal{G}$ be a family of graphs closed under minors and edge subdivisions, and let $\alpha$ be a parameter. Then
    \[
    \alpha_{\mathcal{M}(\mathcal{G})} \leq \alpha 
    \quad \iff \quad 
    \forall G \in \mathcal{G},\ \forall w \in \mathbb{N},\ \exists\, SDD\!\left(G,2w,\tfrac{\alpha}{2w}\right).
    \]
\end{Theorem}

\begin{proof}
Forward direction. Assume $\alpha_{\mathcal{M}(\mathcal{G})}\leq \alpha$, and let $G \in \mathcal{G}$ be arbitrary and $w \in \mathbb{N}$. We define a multicut instance $M$ on $G$ and apply Theorem~\ref{Theorem: the existence of uniform random distribution}. Let 
\[
S = \{(u, v) \in V(G) \times V(G) \mid d_{G}(u, v) \geq 2w \}
\]
be the set of source--sink pairs, and assign unit costs to all edges. Define the fractional solution $x$ by setting $x(e) = \tfrac{1}{2w}$ for all $e \in E(G)$. This is easily seen to be a feasible solution to the LP relaxation \eqref{LP of multicut}. By Theorem~\ref{Theorem: the existence of uniform random distribution}, there exists a probability distribution $y$ over feasible multicuts of $M$ such that each edge $e \in E(G)$ is cut with probability at most
\[
\alpha_{\mathcal{M}(\mathcal{G})} \cdot \tfrac{1}{2w} \;\leq\; \tfrac{\alpha}{2w}.
\]
Finally, note that a set of edges is a feasible multicut for this instance if and only if it defines a $2w$-diameter decomposition of $G$. This yields the desired $SDD(G,2w,\tfrac{\alpha}{2w})$.\\

\noindent Backward direction. Assume that for every $G \in \mathcal{G}$ and $w \in \mathbb{N}$ there exists $SDD(G,2w,\tfrac{\alpha}{2w})$. Let $M$ be an arbitrary instance of the minimum multicut problem on $G \in \mathcal{G}$ with edge costs $c : E(G) \to \mathbb{Z}_{+}$ and source--sink pairs $\{(s_i,t_i)\}_{i=1}^k$. We will show that
\[
\frac{\OPT_{IP}(M)}{\OPT_{LP}(M)} \leq \alpha.
\]

\noindent Let $x = \{x_{e}\}_{e \in E(G)}$ be an optimal fractional solution to the LP relaxation \eqref{LP of multicut}. Since LP solutions are rational, we may assume $x_e \in \mathbb{Q}$ for all $e$. Define the support of $x$ as
\[
A = \{ e \in E(G) \mid x_e > 0 \}.
\]
Choose $w \in \mathbb{N}$ such that $2w x_e$ is an integer for every $e \in A$. Construct a new graph $S(G,M)$ by contracting every edge $e \notin A$ and replacing each edge $e \in A$ with a path $P_e$ of length $2wx_e$. Since $\mathcal{G}$ is closed under minors and edge subdivisions, we have $S(G,M) \in \mathcal{G}$.  

\noindent By assumption, there exists $SDD(S(G,M),2w,\tfrac{\alpha}{2w})$. This means there is a probability distribution $y$ over $\mathcal{F}_{2w}(S(G,M))$ (the family of edge sets inducing components of diameter at most $2w-1$) such that
\[
\sum_{\substack{F \in \mathcal{F}_{2w}(S(G,M)) \\ e \in F}} y_F \;\leq\; \tfrac{\alpha}{2w}, \quad \forall e \in E(S(G,M)).
\]

\begin{Claim}\label{Integrality Gap Claim SDD->multicut}
For the multicut instance $M$, we have
\[
\frac{\OPT_{IP}(M)}{\OPT_{LP}(M)} \leq \alpha.
\]
\end{Claim}
\begin{proof}[{Proof of Claim}]
        \renewcommand{\qedsymbol}{$\blacktriangle$}
For each \( F \in \mathcal{F}_{2w}(S(G, M)) \), define
\[
g(F) = \left\{ e \in E(G)\;|\;  F\cap E(P_{e})\neq \varnothing \right\}.
\]

\noindent\( g(F) \) is a feasible multicut for the instance \( M \). The reason is as follows. Suppose, for contradiction, that there exists a source-sink pair \( (s_1, t_1) \) and a path \( P \in \mathcal{P}_1 \) connecting them in \( G\setminus g(F) \) such that \( E(P) \cap g(F) = \varnothing \). Since the LP solution $x$ satisfies \( \sum_{e \in E(P)} x_e \geq 1 \), the corresponding path \( P' \) in \( S(G, M) \) (formed by replacing each \( e \in E(P) \) with the path \( P_e \)) has length at least \( 2w \). The assumption \( E(P) \cap g(F) = \varnothing \) implies that \( E(P') \cap F = \varnothing \), and thus \( s_1 \) and \( t_1 \) remain connected in \( S(G, M) \setminus F \). This contradicts the assumption that \( F \in \mathcal{F}_{2w}(S(G, M)) \), as no connected component in the decomposition can have diameter \( \geq 2w \).

\noindent Now, let \( B = \{ g(F) \mid F \in \mathcal{F}_{2w}(S(G, M)) \} \), and for each \( b \in B \), define
\[
y'_b = \sum_{\substack{F \in \mathcal{F}_{2w}(S(G, M))\\ g(F) = b}} y_F.
\]
Then \( y' = \{ y'_b \}_{b \in B} \) is a probability distribution over multicuts in \( G \). We now analyze the expected cost of a multicut drawn from this distribution. For any edge \( e \in E(G) \), if \( x_e = 0 \), then \( e \) was contracted and does not appear in \( S(G, M) \), so:
\[
\sum_{a \ni e,\; b \in B} y'_b = 0.
\]
If \( e \in A \), then:
\[
\sum_{b \ni e,\; b \in B} y'_b = \sum_{e' \in P_e} \sum_{\substack{F \ni e'\\ F \in \mathcal{F}_{2w}(S(G, M))}} y_F \leq |E(P_e)| \cdot \frac{\alpha}{2w} = \alpha x_e.
\]
Therefore, in expectation, each edge \( e \) appears in a randomly sampled multicut with probability at most \( \alpha x_e \). This implies there exists a multicut \( b \in B \) such that
\[
\sum_{e \in b} c_e \leq \alpha \sum_{e \in E(G)} c_e x_e = \alpha \cdot OPT_{LP}(M),
\]
which completes the proof.
\end{proof}

\end{proof}

\noindent The transition to the SDD framework eliminates the dependence on the specific placement of source--sink pairs and edge costs, which could otherwise be arbitrary, and instead provides a uniform way of analyzing the integrality gap. Theorem~\ref{relation integrality gap and SDD} serves as the bridge between SDDs and the integrality gap. We only make use of the forward direction of Theorem~\ref{relation integrality gap and SDD} in the following.

\subsection{Small Diameter Decomposition for Trees}

\noindent As mentioned earlier, \( \alpha_{\mathcal{M}(\textsc{tree})} = 2 \). By Theorem~\ref{relation integrality gap and SDD}, this implies that for any tree $T$ and any integer $w \in \mathbb{N}$, there exists 
\[
SDD\!\left(T, 2w, \tfrac{2}{2w} = \tfrac{1}{w}\right).
\]
In other words, the family of trees is $SDD(2w,\tfrac{1}{w})$-acceptable for any $w\in \mathbb{N}$. Moreover, without directly appealing to Theorem~\ref{relation integrality gap and SDD}, we can explicitly construct such an $SDD(T,2w,\tfrac{1}{w})$. This explicit construction will serve as a foundation for developing intuition regarding structured small-diameter decompositions in the next \cref{section: A Structural Result for Small Diameter Decompositions}.

\begin{Theorem}\label{Theorem: trees-simple}
Let \( T \) be a tree. Then for every integer \( w \in \mathbb{N} \), there exists 
\[
SDD\!\left(T,2w,\tfrac{1}{w}\right).
\]
Equivalently, there exists a probability distribution $\mathcal{D} = \{y_{F}\}_{F \in \mathcal{F}_{2w}(T)}$ over $\mathcal{F}_{2w}(T)$ such that
\begin{align}\label{eq: bar edges for trees}
    \sum_{\substack{F \in \mathcal{F}_{2w}(T) \\ e \in F}} y_{F} \;\leq\; \frac{1}{w} \quad \forall e \in E(T).
\end{align}
\end{Theorem}

\begin{proof}
Root the tree $T$ at an arbitrary vertex $r \in V(T)$. For $i = 0, \ldots, w-1$, define
\[
F_i = \{\, e \in E(T) \mid d(r,e) = i + k w \text{ for some } k \in \mathbb{Z}_{\geq 0}\,\}.
\]
Set $y_{F_i} = \tfrac{1}{w}$ for each $i = 0, \ldots, w-1$, and $y_F = 0$ otherwise. Note that the sets $F_i$ partition $E(T)$: we have $E(T) = \bigcup_{i=0}^{w-1} F_i$ and $F_i \cap F_j = \varnothing$ for $i \neq j$. Thus,
\[
\sum_{\substack{F \in \mathcal{F}_{2w}(T) \\ e \in F}} y_F 
= \sum_{\substack{i=0 \\ e \in F_i}}^{w-1} y_{F_i} 
= \tfrac{1}{w}, \quad \forall e \in E(T).
\]

\noindent It remains to show that each $F_i$ is a valid $2w$-diameter decomposition. Fix $F_i$, and consider a pair of vertices $(u,v)$ with $d(u,v) \geq 2w$. Let $q$ be the lowest common ancestor of $u$ and $v$. The unique $u$--$v$ path consists of the $u$--$q$ path and the $q$--$v$ path. Since $d(u,v) \geq 2w$, one of these subpaths has length at least $w$. Without loss of generality, suppose $d(q,v) \geq w$, and denote this path by $Q = e_0, e_1, \dots, e_p$. Because $q$ is an ancestor of $v$, we have $d(r,e_i) = d(r,e_{i-1}) + 1$ for $i=1,\dots,p$. Hence there exists some $e_j \in Q$ such that $d(r,e_j) \equiv i \pmod{w}$, i.e., $e_j \in F_i$. Removing $F_i$ therefore separates $u$ and $v$, as required. This shows that $F_i$ defines a $2w$-diameter decomposition, and hence $\mathcal{D}$ is a valid $SDD(T,2w,1/w)$.
\end{proof}

\noindent The $2w$-diameter decompositions $F_{0},\ldots,F_{w-1}$ described in the proof of Theorem~\ref{Theorem: trees-simple} will be useful in the remainder of the paper, so we record a formal definition.

\begin{Definition}\label{Definition: interesting SDDs trees}
Let $w \in \mathbb{N}$, and let $T$ be a tree with a distinguished root vertex $r \in V(T)$. For each \( i = 0, 1, \ldots, w-1 \), define
\[
    F^{i}_{w}(T,r) := \left\{ e \in E(T) \;\middle|\; d_{T}(r,e) \equiv i \pmod{w} \right\},
\]
where $d_{T}(r,e)$ denotes the distance from $r$ to the closer endpoint of $e$. Then $\{F^{i}_{w}(T,r)\}_{i=0}^{w-1}$ forms a partition of $E(T)$. Moreover, each $F^{i}_{w}(T,r)$ defines a $2w$-diameter decomposition of $T$, and the connected component containing the root $r$ has radius at most $i$ from $r$, that is,
\[
    \mathrm{rad}_{F^{i}_{w}(T,r)}(r) \leq i \leq w-1.
\]

\end{Definition}

\noindent
Thus, the $SDD(T,2w,1/w)$ constructed in the proof of Theorem~\ref{Theorem: trees-simple} also satisfies this useful structural property, which we highlight next.

\begin{Observation}\label{Observation: trees} 
For the $2w$-diameter decompositions $F_{0},\ldots,F_{w-1}$ described in the proof of Theorem~\ref{Theorem: trees-simple}, the following properties hold:
\begin{enumerate}
    \item For every $i = 0,\ldots,w-1$, we have 
    \[
    \rad_{F_i}(r) \leq i \leq w-1.
    \]
    Equivalently,
    \[
    \sum_{F : \rad_{F}(r) \leq w-1} y_F = 1.
    \]

    \item For all $1 \leq k \leq w$, we have $\rad_{F_i}(r) \leq k-1$ with probability $\tfrac{k}{w}$. More precisely,
    \[
    \sum_{F_i : \rad_{F_i}(r) \leq k-1} y_{F_i} 
    = \sum_{i=0}^{k-1} y_{F_i} 
    = \frac{k}{w} 
    \;\;\geq\;\; 1 - \frac{2}{2w}(w-k) 
    = 1 - \frac{\alpha_{\mathcal{M}(\textsc{tree})}}{2w}(w-k).
    \]
\end{enumerate}
\end{Observation}

\noindent This observation shows that the $SDD(T,2w,1/w)$ established in Theorem~\ref{Theorem: trees-simple} not only meets the basic edge condition of equation~\eqref{eq: bar edges for trees}, but also enjoys additional structural properties. These strengthened features are captured in the following Corollary \ref{Corollary: trees- interesting properties}. In the next \cref{section: A Structural Result for Small Diameter Decompositions}, we will extend this idea and show that a similar phenomenon holds for families of graphs closed under the $1$-sum operation, a property that the family of trees also satisfies.

\begin{Corollary}\label{Corollary: trees- interesting properties}
    Let $T$ be a tree, and let $r \in V(T)$ be an arbitrary root. For the family $\mathcal{F}_{2w}(T)$ of all $2w$-diameter decompositions, there exists an $SDD(T,2w,1/w)$, i.e., a distribution $\mathcal{D}=\{y_{F}\}_{F \in \mathcal{F}_{2w}(T)}$, with the following additional properties:
    \begin{align*}
        \sum_{F : \rad_F(r) \leq w-1} y_F &= 1, \\
        \sum_{F : \rad_F(r) \leq k-1} y_F &\geq 1 - \frac{\alpha_{\mathcal{M}(\textsc{tree})}}{2w}(w-k) \quad \forall k = 1,\ldots,w.
    \end{align*}
\end{Corollary}

\section{A Structural Result for Small Diameter Decompositions}\label{section: A Structural Result for Small Diameter Decompositions}
We now define the 1-sum operation on graphs, which will play a crucial role going forward.
\begin{Definition}
    Let \( G_1, \ldots, G_l \) be non-empty graphs, and let \( r_i \in V(G_i) \) for \( i = 1, \ldots, l \). The graph \( G^S \) is obtained by taking the disjoint union of \( G_1, G_2, \ldots, G_l \), and identifying the vertices \( r_1, r_2, \ldots, r_l \). We say that \( G^S \) is obtained by performing the \emph{1-sum} of the \( G_i \)'s at the vertices \( r_i \)'s. The vertex \( r = r_1 = \cdots = r_l \) is called the main vertex of \( G^S \). See the following figure \ref{figure: cut vertex operation} for an illustration:

\begin{figure}[H]
\tikzset{every picture/.style={line width=0.75pt}} 

\begin{tikzpicture}[x=0.75pt,y=0.75pt,yscale=-1,xscale=1]

\draw   (283.93,156.36) .. controls (278.92,146.52) and (291.68,129.98) .. (312.44,119.42) .. controls (333.2,108.86) and (354.09,108.29) .. (359.09,118.13) .. controls (364.1,127.98) and (351.33,144.52) .. (330.58,155.07) .. controls (309.82,165.63) and (288.93,166.21) .. (283.93,156.36) -- cycle ;
\draw   (310.04,186.73) .. controls (301.05,180.3) and (304.75,159.74) .. (318.3,140.8) .. controls (331.84,121.85) and (350.11,111.71) .. (359.09,118.13) .. controls (368.08,124.56) and (364.38,145.12) .. (350.83,164.06) .. controls (337.29,183.01) and (319.02,193.15) .. (310.04,186.73) -- cycle ;
\draw   (423.72,172.3) .. controls (416.63,180.77) and (396.41,175.51) .. (378.56,160.55) .. controls (360.71,145.59) and (352,126.6) .. (359.09,118.13) .. controls (366.19,109.67) and (386.41,114.93) .. (404.26,129.89) .. controls (422.1,144.85) and (430.82,163.84) .. (423.72,172.3) -- cycle ;
\draw   (60.81,191) .. controls (49.77,190.94) and (40.9,175.22) .. (41,155.89) .. controls (41.1,136.56) and (50.14,120.94) .. (61.19,121) .. controls (72.23,121.06) and (81.1,136.78) .. (81,156.11) .. controls (80.9,175.44) and (71.86,191.06) .. (60.81,191) -- cycle ;
\draw   (107.81,191) .. controls (96.77,190.94) and (87.9,175.22) .. (88,155.89) .. controls (88.1,136.56) and (97.14,120.94) .. (108.19,121) .. controls (119.23,121.06) and (128.1,136.78) .. (128,156.11) .. controls (127.9,175.44) and (118.86,191.06) .. (107.81,191) -- cycle ;
\draw   (182.81,191) .. controls (171.77,190.94) and (162.9,175.22) .. (163,155.89) .. controls (163.1,136.56) and (172.14,120.94) .. (183.19,121) .. controls (194.23,121.06) and (203.1,136.78) .. (203,156.11) .. controls (202.9,175.44) and (193.86,191.06) .. (182.81,191) -- cycle ;
\draw   (218,141.75) -- (244.6,141.75) -- (244.6,137) -- (262.33,146.5) -- (244.6,156) -- (244.6,151.25) -- (218,151.25) -- cycle ;
\draw  [fill={rgb, 255:red, 0; green, 0; blue, 0 }  ,fill opacity=1 ] (58.19,121) .. controls (58.19,119.34) and (59.53,118) .. (61.19,118) .. controls (62.85,118) and (64.19,119.34) .. (64.19,121) .. controls (64.19,122.66) and (62.85,124) .. (61.19,124) .. controls (59.53,124) and (58.19,122.66) .. (58.19,121) -- cycle ;
\draw  [fill={rgb, 255:red, 0; green, 0; blue, 0 }  ,fill opacity=1 ] (180.19,121) .. controls (180.19,119.34) and (181.53,118) .. (183.19,118) .. controls (184.85,118) and (186.19,119.34) .. (186.19,121) .. controls (186.19,122.66) and (184.85,124) .. (183.19,124) .. controls (181.53,124) and (180.19,122.66) .. (180.19,121) -- cycle ;
\draw  [fill={rgb, 255:red, 0; green, 0; blue, 0 }  ,fill opacity=1 ] (104.19,121) .. controls (104.19,119.34) and (105.53,118) .. (107.19,118) .. controls (108.85,118) and (110.19,119.34) .. (110.19,121) .. controls (110.19,122.66) and (108.85,124) .. (107.19,124) .. controls (105.53,124) and (104.19,122.66) .. (104.19,121) -- cycle ;
\draw  [fill={rgb, 255:red, 0; green, 0; blue, 0 }  ,fill opacity=1 ] (356.09,118.13) .. controls (356.09,116.47) and (357.44,115.13) .. (359.09,115.13) .. controls (360.75,115.13) and (362.09,116.47) .. (362.09,118.13) .. controls (362.09,119.79) and (360.75,121.13) .. (359.09,121.13) .. controls (357.44,121.13) and (356.09,119.79) .. (356.09,118.13) -- cycle ;

\draw (55,101.4) node [anchor=north west][inner sep=0.75pt]    {$r_{1}$};
\draw (102,101.4) node [anchor=north west][inner sep=0.75pt]    {$r_{2}$};
\draw (177,101.4) node [anchor=north west][inner sep=0.75pt]    {$r_{k}$};
\draw (134,150.4) node [anchor=north west][inner sep=0.75pt]  [font=\Large]  {$...$};
\draw (53,147.4) node [anchor=north west][inner sep=0.75pt]    {$G_{1}$};
\draw (99,147.4) node [anchor=north west][inner sep=0.75pt]    {$G_{2}$};
\draw (173,146.4) node [anchor=north west][inner sep=0.75pt]    {$G_{l}$};
\draw (288,137.4) node [anchor=north west][inner sep=0.75pt]    {$G_{1}$};
\draw (311,164.4) node [anchor=north west][inner sep=0.75pt]    {$G_{2}$};
\draw (393.41,148.62) node [anchor=north west][inner sep=0.75pt]    {$G_{l}$};
\draw (356,103.4) node [anchor=north west][inner sep=0.75pt]    {$r=r_{1},\cdots,r_{l}$};
\end{tikzpicture}
    \centering
    \caption{An illustration of the 1-sum operation}
    \label{figure: cut vertex operation}
\end{figure}
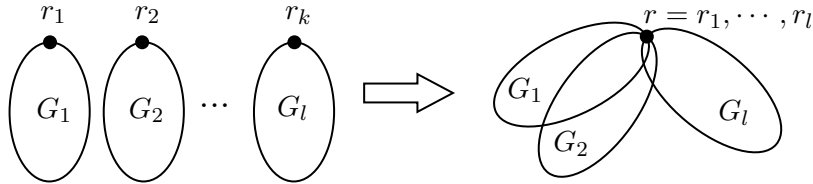

\end{Definition}

Let $\mathcal{G}$ be a family of graphs. We say that $\mathcal{G}$ is closed under the 1-sum operation if for any $G_{1},\ldots,G_{l}\in \mathcal{G}$ and $r_i \in G_i$, the graph obtained by taking 1-sum of $G_{1},\ldots,G_{l}$ at $r_1,r_2,\ldots,r_l$ is a graph in $\mathcal{G}$. Many natural classes of family are closed under the 1-sum operation, such as trees, cactus graphs and planar graphs. Note that 1-sum is a special case of a well known and a more general notion of clique-sums. 
  
\noindent In the remainder of this section, let $w \in \mathbb{N}$ and $0 < p < 1$ be fixed parameters, and assume that the graph family $\mathcal{G}$ is closed under minors, subdivisions, and the $1$-sum operation, and is $SDD(2w,p)$-acceptable. Note that the existence of an $SDD(G,2w,p)$ implies that one can sample a $2w$-diameter decomposition of $G$ in which each edge is included with probability at most $p$.\\

\noindent We note down a few more definitions before stating the main theorems of this section. Let $G$ be a graph and $r\in V(G)$ be an arbitrary vertex. Recall that $\mathcal{F}_{2w}(G)$ denotes the set of all $2w$-diameter decompositions of $G$. For $k \in \{1,\ldots,w\}$, we use $\mathcal{F}^{k}_{2w}(G,r)$ to denote the set of all $2w$-diameter decompositions of $G$ such that every vertex in the connected component containing $r$ is within distance strictly less than $k$ from it. More precisely, 
$$\mathcal{F}^{k}_{2w}(G,r)=\{F\in \mathcal{F}_{2w}(G)~|~rad_{F}(r) < k\}.$$

\noindent We now state a simple Lemma \ref{Lemma: projection is pload}, which will be used in the proofs of Theorems 
\ref{Theorem: rooted (p,2w)-diameter distribution}, 
\ref{Theorem: radius rooted (p,2w)-diameter distribution}, and 
\ref{Theorem: generalized radius rooted (p,2w)-diameter distribution}.

\begin{Lemma}\label{Lemma: projection is pload}
Let \( G \) be a graph and let \( x \) be a $SDD(G,2w,p)$ for \( G \). 
If \( H \) is a subgraph of \( G \), then the distribution
\[
y_{F} = \sum_{\substack{F' \in \mathcal{F}_{2w}(G) \\ F' \cap E(H) = F}} x_{F'}
\quad \text{for each } F \in \mathcal{F}_{2w}(H),
\]
is a $SDD(H,2w,p)$ for \( H \).
\end{Lemma}

\begin{proof}
Fix an edge \( e \in E(H) \). Then
\[
\sum_{e \in F} y_{F} 
  = \sum_{e \in F} \ \sum_{\substack{F' \in \mathcal{F}_{2w}(G) \\ F' \cap E(H) = F}} x_{F'} 
  = \sum_{\substack{F' \in \mathcal{F}_{2w}(G) \\ e \in F'}} x_{F'} 
  \leq p,
\]
where the inequality follows from the fact that \( x \) is a $SDD(G,2w,p)$ for \( G \). Moreover, since \( y \geq 0 \), we verify that \( y \) forms a probability distribution:
\[
\sum_{F \in \mathcal{F}_{2w}(H)} y_{F} 
  = \sum_{F \in \mathcal{F}_{2w}(H)} \ \sum_{\substack{F' \in \mathcal{F}_{2w}(G) \\ F' \cap E(H) = F}} x_{F'}
  = \sum_{F' \in \mathcal{F}_{2w}(G)} x_{F'} 
  = 1.
\]
Thus, \( y \) is a $SDD(H,2w,p)$ for \( H \).
\end{proof}

\begin{Definition}[Projection of a $SDD$]\label{def:SDD-projection}
 Let \( G \) be a graph and let \( x \) be a $SDD(G,2w,p)$ for \( G \). For a subgraph \( H \subseteq G \), the projection of \( x \) onto \( H \) is the distribution
\[
x(H) = \big( y_{F} \big)_{F \in \mathcal{F}_{2w}(H)},
\]
where for each \( F \in \mathcal{F}_{2w}(H) \),
\[
y_{F} = \sum_{\substack{F' \in \mathcal{F}_{2w}(G) \\ F' \cap E(H) = F}} x_{F'}.
\]
By Lemma~\ref{Lemma: projection is pload}, \( x(H) \) is a $SDD(H,2w,p)$ for \( H \).
\end{Definition}

\noindent In Theorem~\ref{Theorem: rooted (p,2w)-diameter distribution}, we show that if $\mathcal{G}$ is closed under the $1$-sum operation and is $SDD(2w,p)$-acceptable, then for any \( G \in \mathcal{G} \) and \( r \in V(G) \), there exists an $SDD(G,2w,p)$ over the family of \( 2w \)-diameter decompositions \( \mathcal{F}_{2w}(G) \) such that, when sampling a decomposition \( F \in \mathcal{F}_{2w}(G) \) from this distribution, we are guaranteed that \( \rad_F(r) \leq w - 1 \), i.e., \( F \in \mathcal{F}_{2w}^{w}(G,r) \). This condition is directly analogous to the first item of Observation~\ref{Observation: trees}.

\begin{Theorem}\label{Theorem: rooted (p,2w)-diameter distribution}
    Suppose that $\mathcal{G}$ is closed under the $1$-sum operation and is $SDD(2w,p)$-acceptable. 
    Let $G \in \mathcal{G}$ and let $r \in V(G)$ be an arbitrary vertex. 
    Then there exists an $SDD(G,2w,p)$, i.e., a distribution 
    $y = \{ y_{F} \mid F \in \mathcal{F}_{2w}(G) \}$, such that
    \[
        \sum_{F \in \mathcal{F}^{w}_{2w}(G,r)} y_{F} = 1,
    \]
    meaning that every sampled $2w$-diameter decomposition $F$ from this distribution 
    satisfies $\rad_{F}(r) \leq w-1$.
\end{Theorem}

\begin{proof}
Let $\mathcal{F}_{2w}=\mathcal{F}_{2w}(G)$ and $\mathcal{F}_{2w}^{w}(r)=\mathcal{F}_{2w}^{w}(G,r)$ for simplicity. It is sufficient to show that the following LP \eqref{LP SDD first theorem} is feasible and has optimal value $0$.

\begin{equation}\label{LP SDD first theorem}
    \begin{aligned}
    \min \sum_{F \in \mathcal{F}_{2w} \setminus \mathcal{F}^{w}_{2w}(r)} y_F & \\
    \sum_{\substack{F \in \mathcal{F}_{2w} \\ e \in F}} y_F \leq & ~p \quad \forall e \in E(G) \\
    \sum_{F \in \mathcal{F}_{2w}} y_F = & 1 \\
    y_F \geq & 0 \quad \forall F \in \mathcal{F}_{2w}
\end{aligned}
\end{equation}

\noindent The above LP \eqref{LP SDD first theorem} is feasible since $\mathcal{G}$ is $SDD(2w,p)$-acceptable. For the sake of contradiction, assume that the optimal value of the above LP is $z>0$. Let $m>\frac{1}{z}$ be a natural number. Let $G_{1},\ldots,G_{m}$ be $m$ disjoint copies of $G$ and $r_{i}$ be the vertex of $G_i$ which corresponds to $r$. Let $G'$ be formed by taking 1-sum of $G_{1},\ldots,G_{m}$ at $r_{1},\ldots,r_{m}$. See the following Figure \ref{Illustration for construction of G' in first theorem} for an illustration:

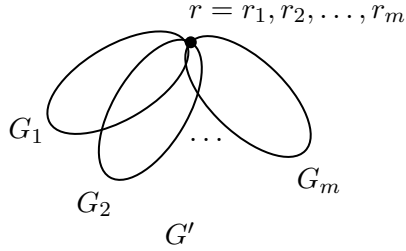
\begin{figure}[H]
    \centering
    \tikzset{every picture/.style={line width=0.75pt}} 
\begin{tikzpicture}[x=0.75pt,y=0.75pt,yscale=-1,xscale=1]

\draw  [fill={rgb, 255:red, 0; green, 0; blue, 0 }  ,fill opacity=1 ] (260,71.5) .. controls (260,70.12) and (261.12,69) .. (262.5,69) .. controls (263.88,69) and (265,70.12) .. (265,71.5) .. controls (265,72.88) and (263.88,74) .. (262.5,74) .. controls (261.12,74) and (260,72.88) .. (260,71.5) -- cycle ;
\draw   (216.61,75.95) .. controls (235.32,64.73) and (254.75,62.74) .. (260,71.5) .. controls (265.25,80.26) and (254.35,96.46) .. (235.64,107.68) .. controls (216.93,118.9) and (197.5,120.89) .. (192.25,112.13) .. controls (186.99,103.37) and (197.9,87.17) .. (216.61,75.95) -- cycle ;
\draw   (226.41,95.99) .. controls (237.56,77.24) and (253.72,66.28) .. (262.5,71.5) .. controls (271.28,76.72) and (269.36,96.16) .. (258.21,114.91) .. controls (247.06,133.66) and (230.9,144.62) .. (222.12,139.4) .. controls (213.34,134.18) and (215.26,114.74) .. (226.41,95.99) -- cycle ;
\draw   (278.3,112.16) .. controls (262.52,97.09) and (255.44,78.89) .. (262.5,71.5) .. controls (269.56,64.11) and (288.07,70.33) .. (303.85,85.39) .. controls (319.63,100.46) and (326.7,118.66) .. (319.64,126.05) .. controls (312.59,133.44) and (294.08,127.22) .. (278.3,112.16) -- cycle ;

\draw (170,108.4) node [anchor=north west][inner sep=0.75pt]    {$G_{1}$};
\draw (204,145.4) node [anchor=north west][inner sep=0.75pt]    {$G_{2}$};
\draw (314,133.4) node [anchor=north west][inner sep=0.75pt]    {$G_{m}$};
\draw (260.21,118.31) node [anchor=north west][inner sep=0.75pt]    {$\ldots$};
\draw (260,51.4) node [anchor=north west][inner sep=0.75pt]    {$r=r_{1} ,r_{2} ,\ldots,r_{m}$};
\draw (248,160.4) node [anchor=north west][inner sep=0.75pt]    {$G'$};
\end{tikzpicture}
    \label{Illustration for construction of G' in first theorem}
\end{figure}

\noindent Note that \( G' \in \mathcal{G} \) since \( \mathcal{G} \) is closed under the 1-sum operation. Let \( \mathcal{F}_{2w}' =\mathcal{F}_{2w}(G')\) be the set of all \( 2w \)-diameter decompositions of \( G' \). Since $\mathcal{G}$ is $SDD(2w,p)$-acceptable, then $G'$ has a $SDD(G',2w,p)$. Let \( \{ g_{F'} \}_{F' \in \mathcal{F}_{2w}(G')} \) denote such a distribution for \( G' \).
 Let \( G_i = (V_i, E_i) \) and \( \mathcal{F}_{2w}(G_{i}) \) be the set of all \( 2w \)-diameter decompositions of \( G_i \) for \( i = 1, 2, \dots, m \).\footnote{Note that \( \bigcap_{i=1}^{m} V_i = \{ r \} \).} By Lemma~\ref{Lemma: projection is pload}, the projection of \( g \) onto \( G_{i} \) induces a distribution \( g^{i} = g(G_{i}) \) over \( \mathcal{F}_{2w}(G_{i}) \) for $i=1,\ldots,m$. Furthermore, since \( G_i \) is an identical copy of \( G \) and \( g^i \) is a feasible solution to the LP \ref{LP SDD first theorem} mentioned above, we have:
\[
\sum_{F \in \mathcal{F}_{2w}(G_i) \setminus \mathcal{F}_{2w}^{w}(G_i, r)} g^i_F \geq z \quad \text{for} \quad i = 1, 2, \dots, m.
\]

\noindent Recall that \( \mathcal{F}_{2w}^{w}(G_i, r) \) denotes the set of all \( 2w \)-diameter decompositions of \( G_i \) in which the distance of every vertex in the connected component containing \( r \) is at most \( w-1 \) from it. Let \( T_i \) be the event that, when sampling \( F' \in \mathcal{F}_{2w}' \) according to the distribution \( g \), the intersection \( F' \cap E_i \) does not belong to \( \mathcal{F}_{2w}^{w}(G_i, r) \). From the above discussion, it follows that \( \text{Pr}[T_i] \geq z \). Since \( m > \frac{1}{z} \), we have
\[
\sum_{i=1}^{m} \text{Pr}[T_i] \geq z \cdot m > 1.
\]
This implies that the events \( T_1, \dots, T_m \) are not disjoint, and there exist indices \( i, j \) such that \( \text{Pr}[T_i \cap T_j] >0 \). Therefore, there exists a \( F' \in \mathcal{F}_{2w}' \), and vertices \( u \in V_i \), \( v \in V_j \) such that:
\begin{enumerate}
    \item \( g_{F'} > 0 \),
    \item \( u \) and \( v \) are in the connected component containing \( r \) in \( G' \setminus F' \),
    \item the distance of \( u \) and \( v \) from \( r \) is at least \( w \).
\end{enumerate}
But then, the diameter of \( F' \) is at least \( 2w \), which contradicts the fact that \( F' \in \mathcal{F}_{2w}' \). This implies that $z=0$, and completes the proof of the theorem.

\end{proof}

\noindent We say that a graph $G$ together with a vertex $r \in V(G)$ has an $SDD(G,2w,p,r)$ if there exists an $SDD(G,2w,p)$ distribution $y = \{ y_F \}_{F \in \mathcal{F}_{2w}(G)}$ such that
\[
\sum_{F \in \mathcal{F}^{w}_{2w}(G,r)} y_F \;=\; 1.
\]
In this case, every sampled $2w$-diameter decomposition $F$ from $y$ satisfies $\rad_F(r) \leq w-1$.  
We say that a graph class $\mathcal{G}$ is strongly $SDD(2w,p)$-acceptable if for every $G \in \mathcal{G}$ and every $r \in V(G)$, there exists an $SDD(G,2w,p,r)$. 

\noindent In Theorem~\ref{Theorem: rooted (p,2w)-diameter distribution}, we proved that if a graph class $\mathcal{G}$, closed under the $1$-sum operation, is $SDD(2w,p)$-acceptable, then it is also strongly $SDD(w,p)$-acceptable. In Theorem~\ref{Theorem: radius rooted (p,2w)-diameter distribution}, we extend this result to arbitrary radii. For the proof, we additionally assume that $\mathcal{G}$ contains $K_2$, the complete graph on two vertices (i.e., a single edge).

\begin{Theorem}\label{Theorem: radius rooted (p,2w)-diameter distribution}
Suppose that $\mathcal{G}$ is closed under the $1$-sum operation, contains $K_2$, and is strongly $SDD(2w,p)$-acceptable.  
Let $G \in \mathcal{G}$, $r \in V(G)$, and $k \in \{1,\ldots,w\}$.  
Then there exists an $SDD(G,2w,p,r)$ distribution $y = \{y_{F} \mid F \in \mathcal{F}_{2w}(G)\}$ such that
\[
\sum_{F \in \mathcal{F}^{k}_{2w}(G,r)} y_{F} \;\geq\; 1 - p(w-k).
\]
\end{Theorem}

\begin{proof}
Let $\mathcal{F}_{2w}=\mathcal{F}_{2w}(G), \mathcal{F}_{2w}^{w}(r)=\mathcal{F}_{2w}^{w}(G,r)$ and $\mathcal{F}_{2w}^{k}(r)=\mathcal{F}_{2w}^{k}(G,r)$ for simplicity. It is sufficient to show that the following LP \eqref{LP multicut second theorem} is feasible and has optimal value $0$.

\begin{equation}\label{LP multicut second theorem}
    \begin{aligned}
     \min\sum_{F\in \mathcal{F}_{2w}\setminus \mathcal{F}_{2w}^{w}(r)} y_{F}&\\
     \sum_{F\in \mathcal{F}_{2w}^{k}(r)}y_{F}\geq &1-p(w-k)\\
    \sum_{\substack{F\in \mathcal{F}_{2w}\\e\in F}} y_{F}\leq &p \quad\forall e\in E(G)\\
    \sum_{F\in \mathcal{F}_{2w}} y_{F}=&1\\
    y_{F}\geq &0\quad\forall F\in \mathcal{F}_{2w}
\end{aligned} 
\end{equation}

\noindent For now, assume that the LP \eqref{LP multicut second theorem} is feasible and its optimal value is $z>0$. The proof that LP \eqref{LP multicut second theorem} is feasible will also follow from the discussion below. Let $m>\frac{1}{z}$ and $G_{1},\ldots,G_{m}$ be $m$ disjoint copies of $G$. Let $r_{i}$ be the vertex of $G_i$ which corresponds to $r$ and $G'$ be formed by taking 1-sum of $G_{1},\ldots,G_{m}$ at $r_{1},\ldots,r_{m}$. We construct $H$ by adding a path of length $w-k$ to $G'$ at $r$. Let $P=\{e_1,e_2,\ldots,e_{w-k}\}$ be the set of edges on this path, where $e_1=(r,v_1),e_2=(v_1,v_2),\ldots,e_{w-k}=(v_{w-k-1},v)$. See the following Figure \ref{H in the second theorem} for an illustration:

\begin{figure}[H]
    \centering
    \tikzset{every picture/.style={line width=0.75pt}} 

\begin{tikzpicture}[x=0.75pt,y=0.75pt,yscale=-1,xscale=1]

\draw  [fill={rgb, 255:red, 0; green, 0; blue, 0 }  ,fill opacity=1 ] (284,129.5) .. controls (284,128.12) and (285.12,127) .. (286.5,127) .. controls (287.88,127) and (289,128.12) .. (289,129.5) .. controls (289,130.88) and (287.88,132) .. (286.5,132) .. controls (285.12,132) and (284,130.88) .. (284,129.5) -- cycle ;
\draw   (240.61,133.95) .. controls (259.32,122.73) and (278.75,120.74) .. (284,129.5) .. controls (289.25,138.26) and (278.35,154.46) .. (259.64,165.68) .. controls (240.93,176.9) and (221.5,178.89) .. (216.25,170.13) .. controls (210.99,161.37) and (221.9,145.17) .. (240.61,133.95) -- cycle ;
\draw   (250.41,153.99) .. controls (261.56,135.24) and (277.72,124.28) .. (286.5,129.5) .. controls (295.28,134.72) and (293.36,154.16) .. (282.21,172.91) .. controls (271.06,191.66) and (254.9,202.62) .. (246.12,197.4) .. controls (237.34,192.18) and (239.26,172.74) .. (250.41,153.99) -- cycle ;
\draw   (302.3,170.16) .. controls (286.52,155.09) and (279.44,136.89) .. (286.5,129.5) .. controls (293.56,122.11) and (312.07,128.33) .. (327.85,143.39) .. controls (343.63,158.46) and (350.7,176.66) .. (343.64,184.05) .. controls (336.59,191.44) and (318.08,185.22) .. (302.3,170.16) -- cycle ;
\draw  [fill={rgb, 255:red, 0; green, 0; blue, 0 }  ,fill opacity=1 ] (283,104.5) .. controls (283,103.12) and (284.12,102) .. (285.5,102) .. controls (286.88,102) and (288,103.12) .. (288,104.5) .. controls (288,105.88) and (286.88,107) .. (285.5,107) .. controls (284.12,107) and (283,105.88) .. (283,104.5) -- cycle ;
\draw  [fill={rgb, 255:red, 0; green, 0; blue, 0 }  ,fill opacity=1 ] (283,67.5) .. controls (283,66.12) and (284.12,65) .. (285.5,65) .. controls (286.88,65) and (288,66.12) .. (288,67.5) .. controls (288,68.88) and (286.88,70) .. (285.5,70) .. controls (284.12,70) and (283,68.88) .. (283,67.5) -- cycle ;
\draw  [fill={rgb, 255:red, 0; green, 0; blue, 0 }  ,fill opacity=1 ] (283,44.5) .. controls (283,43.12) and (284.12,42) .. (285.5,42) .. controls (286.88,42) and (288,43.12) .. (288,44.5) .. controls (288,45.88) and (286.88,47) .. (285.5,47) .. controls (284.12,47) and (283,45.88) .. (283,44.5) -- cycle ;
\draw    (285.5,44.5) -- (285.5,70) ;
\draw    (285.5,104.5) -- (286.5,127) ;

\draw (194,166.4) node [anchor=north west][inner sep=0.75pt]    {$G_{1}$};
\draw (228,203.4) node [anchor=north west][inner sep=0.75pt]    {$G_{2}$};
\draw (338,191.4) node [anchor=north west][inner sep=0.75pt]    {$G_{m}$};
\draw (284.21,176.31) node [anchor=north west][inner sep=0.75pt]    {$\ldots$};
\draw (290,113.4) node [anchor=north west][inner sep=0.75pt]    {$r=r_{1}=r_{2}= \ldots=r_{m}$};
\draw (278,224.4) node [anchor=north west][inner sep=0.75pt]    {$H$};
\draw (282,25.4) node [anchor=north west][inner sep=0.75pt]    {$v$};
\draw (283.57,93.88) node [anchor=north west][inner sep=0.75pt]  [rotate=-270.23]  {$\ldots$};

\end{tikzpicture}
    \caption{Illustration of the construction of $H$}
    \label{H in the second theorem}
\end{figure}
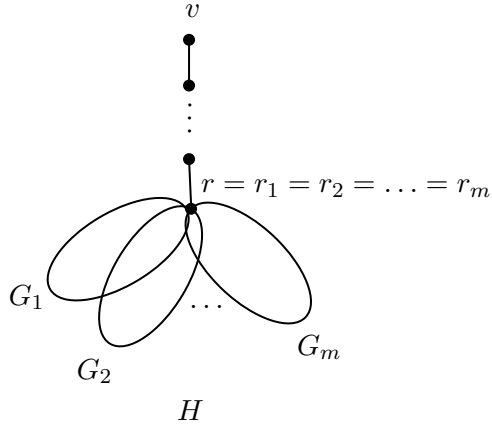

\noindent Since $\mathcal{G}$ is closed under the $1$-sum operation and $K_2 \in \mathcal{G}$, we conclude that $H \in \mathcal{G}$. 
By the assumption that $\mathcal{G}$ is strongly $SDD(2w,p)$-acceptable, there exists an $SDD(H,2w,p,r)$ distribution 
\[
x = \{x_{F'} \mid F' \in \mathcal{F}_{2w}(H)\},
\]
with properties
\[
\sum_{F' \in \mathcal{F}^{w}_{2w}(H,r)} x_{F'} = 1
\quad \text{and} \quad
\sum_{\substack{F' \in \mathcal{F}_{2w}(H) \\ e \in F'}} x_{F'} \leq p \quad \text{for all } e \in E(H).
\]

\noindent Let $A=\{F' \in \mathcal{F}^{w}_{2w}(H,v)~|~F' \cap E(P) \neq \varnothing \}$ and $B=\{F' \in \mathcal{F}^{w}_{2w}(H,v)~|~F' \cap E(P) =\varnothing \}$. Let $A_{i}=\{F' \in  A~|~e_{i}\in F'\}$ for $i=1,\ldots,w-k$. Note that $\sum_{F' \in A}x_{F'}+\sum_{F' \in B}x_{F'}=1$. Since $\sum_{F' \in A_{i}} x_{F'} \leq p$, and $A=\cup_{i=1}^{w-k}A_{i}$, it follows that: 

$$\sum_{F'\in A} x_{F'} \leq (w-k)p ~~\Rightarrow~~\sum_{F' \in B}x_{F'}\geq 1-(w-k)p.$$

\noindent Let \( G_i = (V_i, E_i) \) and \( \mathcal{F}_{2w}(G_i) \) be the set of all \( 2w \)-diameter decompositions of \( G_i \) for \( i= 1, 2, \dots, m \). Recall that \( \mathcal{F}_{2w}^{w}(G_i, r),\mathcal{F}_{2w}^{k}(G_i, r) \) is the set of \( 2w \)-diameter decompositions of \( G_i \) in which the distance of every vertex in the connected component containing \( r \) is at most $w-1$ and $k-1$ from it, respectively. The next claim shows that the projection of the distribution \( x \) onto \( G_i \), 
denoted \(y^{i} = x(G_{i}) \), yields a feasible solution to the LP~\eqref{LP multicut second theorem}.

\begin{Claim}\label{claim: y_j feasible}
The distribution \( y^{i} \) is a feasible solution to the 
LP~\eqref{LP multicut second theorem}.
\end{Claim}

\begin{proof}[{Proof of Claim}]
        \renewcommand{\qedsymbol}{$\blacktriangle$}
 Lemma \ref{Lemma: projection is pload} implies $y^{i}$ is a $SDD(G_{i},2w,p)$ for $G_{i}$. Let $B_{i}=\{F'\cap E_i~|~F'\in B\}$. Observe that $B_{i}\in \mathcal{F}_{2w}(G_i)$. Let $F\in B_{i}$. Then there exists $F'\in B$ such that $F=F'\cap E_i$. Furthermore, for any $F' \in B$, we have $F' \in \mathcal{F}_{2w}^{w}(H,v)$ and $F'\cap E(P)=\varnothing$. Hence we can conclude that $F \in \mathcal{F}^{k}_{2w}(G_i,r)$. Thus, $$\sum_{F\in \mathcal{F}^{k}_{2w}(G_i,r)} y^{i}_{F} \geq \sum_{F\in B_{i}} y^{i}_{F}=\sum_{F\in B_{i}}~~\sum_{\substack{F'\in B\\F'\cap E_i=F}} x_{F'}=\sum_{F'\in B} x_{F'}\geq 1-(w-k) \cdot p.$$

\end{proof}

\noindent Since $G_i$ is an identical copy of $G$ and $\mathcal{F}_{2w}(G_{i})$ is also an identical copy of $\mathcal{F}_{2w}=\mathcal{F}_{2w}(G)$, Claim~\ref{claim: y_j feasible} shows that $y^{i}$ is a feasible solutions for LP~\eqref{LP multicut second theorem} for $G_i$, it follows that LP~\eqref{LP multicut second theorem} is feasible. Since $z$ is the optimal solution of LP~\eqref{LP multicut second theorem}, for each $y^{i}$, we have:

$$\sum_{F\in \mathcal{F}_{2w}(G_i)\setminus \mathcal{F}_{2w}^{w}(G_i,r)} y^{i}_{F}\geq z.$$

\noindent This implies that,
$$\sum_{i=1}^{m}~~\sum_{F\in \mathcal{F}_{2w}(G_i)\setminus \mathcal{F}_{2w}^{w}(G_i,r)} y^{i}_{F} \geq mz >1.$$ 
Using same argument as in the proof of Theorem~\ref{Theorem: rooted (p,2w)-diameter distribution}, we can show that there exists $1\leq i\neq j\leq m$ and $F' \in \mathcal{F}_{2w}^{w}(H,v)$ such that $$y_{F'} >0,~F' \cap E_i \notin \mathcal{F}^{w}_{2w}(G_{i},r)~\text{and}~F' \cap E_j \notin \mathcal{F}^{w}_{2w}(G_j,r).$$ This means that there exists a vertex $a \in V_i,b\in V_j$ such that the distance of $a$ and $b$ from $r$ is  at least $w$, and they are both included in the connected component containing $r$ in $H-F'$. Hence $u$ and $v$ are at least $2w$ distance apart, and this contradicts the fact that $F'$ is a $2w$-diameter decomposition. Hence $z=0$ and this completes the proof of the theorem.

\end{proof}

So far, we have shown that SDDs for any graph class closed under the $1$-sum operation are consistent with those constructed for trees in Observation \ref{Observation: trees}. We now state a simple Lemma~\ref{lemma:upperbound}, which will be useful in proving the lower bound on the integrality gap for cactus graphs in \cref{sec:proof_main_thm}. From this point onward, by Theorem~\ref{Theorem: rooted (p,2w)-diameter distribution} and Theorem~\ref{Theorem: radius rooted (p,2w)-diameter distribution}, we may restrict our attention to $\mathcal{F}_{2w}^{w}(G,r)$ instead of $\mathcal{F}_{2w}(G)$, for a given graph $G$ and vertex $r\in V(G)$.
\begin{Lemma} \label{lemma:upperbound}
    Let $G\in \mathcal{G}$ and let $r\in V(G)$ be an arbitrary vertex.  
    Let $P$ be any shortest path of length $w$ starting at $r$, and denote its other endpoint by $r'$.  
    If $y=\{y_{F}~|~F\in \mathcal{F}_{2w}^{w}(G,r)\}$ is an $SDD(G,2w,p,r)$, then:
    \[
    F\cap E(P)\neq \varnothing \quad \text{for all } F\in \mathcal{F}_{2w}^{w}(G,r),
    \]
    and
    \[
    \sum_{\substack{F \in \mathcal{F}_{2w}^{w}(G,r) \\ |F\cap E(P)|\geq 2}} y_{F}\;\leq\; (w \cdot p)-1.
    \]
\end{Lemma}

\begin{proof}
    If there exists $F\in \mathcal{F}_{2w}^{w}(G,r)$ such that $F \cap E(P)=\varnothing$, then $r,r'$ are within the same connected component in $G- F$, which contradicts the fact that $F \in \mathcal{F}_{2w}^{w}(G,r)$. Let, 
    $$A=\{F\in \mathcal{F}_{2w}^{w}(G,r)~|~|F\cap E(P)|\geq 2 \}~\text{and}~B=\{F\in \mathcal{F}_{2w}^{w}(G,r)~|~|F\cap E(P)|=1\}.$$ 
    Note that $A,B$ forms a partition of $\mathcal{F}_{2w}^{w}(G,r)$. Using the definition of $A$ and $B$, and the fact that $\sum_{F\in A} y_{F} + \sum_{F\in B} y_{F}=1$, we can derive the statement of the theorem as follows:
    $$\sum_{F\in A} y_{F}+1= 2 \cdot \sum_{F\in A} y_{F} + \sum_{F\in B} y_{F} \leq \sum_{e\in E(P)}~~\sum_{\substack{F\in \mathcal{F}_{2w}^{w}(G,r)\\e\in F}} y_{F} \leq~\sum_{e\in E(P)}p~\leq~w \cdot p.$$ 
    Note that the first inequality can be derived by showing that $y_{F}$ appears at least twice in the right hand side if $F\in A$, and exactly once if $F\in B$.\end{proof}

\noindent We are now equipped to prove the lower bound on the integrality gap for cactus graphs, which will be the focus of the next \cref{sec:proof_main_thm}.

\section{$\frac{20}{9}$ Lower Bound For Cactus Graphs}\label{sec:proof_main_thm}

\noindent We are now ready to prove our first theorem. Suppose that $w$ is a fixed even integer. We will apply the tools developed in the previous sections to the family of cactus graphs.
 Let \( \mathcal{G} \) be the family of cactus graphs, and define \( \alpha = \alpha_{\mathcal{M}(\mathcal{G})} \). Since $\mathcal{G}$ is closed under minors and subdivisions, Theorem~\ref{relation integrality gap and SDD} implies that $\mathcal{G}$ is $SDD\!\left(2w,\tfrac{\alpha}{2w}\right)$-acceptable. Let \( p = \tfrac{\alpha}{2w} \). Since \( \mathcal{G} \) is also closed under the 1-sum operation, Theorem~\ref{Theorem: rooted (p,2w)-diameter distribution} ensures that \( \mathcal{G} \) is strongly $SDD(2w,p)$-acceptable. In Theorem~\ref{Theorem Cactus 20/9}, we show that if \( \mathcal{G} \) is strongly $SDD(2w,p)$-acceptable, then \( w \cdot p \geq \tfrac{10}{9} \). This implies
\[
w \cdot p \;=\; w \cdot \frac{\alpha}{2w} \;\geq\; \frac{10}{9},
\]
and hence $\alpha \geq \tfrac{20}{9}$.

\begin{Theorem}\label{Theorem Cactus 20/9}
     Let $\mathcal{G}$ be the family of cactus graphs. If $\mathcal{G}$ is strongly $SDD(2w,p)$-acceptable, then $w \cdot p \geq \tfrac{10}{9}$.
\end{Theorem}  

\begin{proof}
Let $H$ be a cycle of length $2w$. Let $r$, $u$, $r'$, and $v$ be four vertices of the cycle in anti-clockwise order, such that $d(r,u) = d(u,r') = d(r',v) = d(v,r)=\frac{w}{2}$.  We construct $G$ from $H$ by attaching paths $u-u'$ and $v-v'$ of length $\frac{w}{2}$ from $u$ and $v$, respectively. We denote the path of length $\frac{w}{2}$ from $r$ to $u$ by $P_u$, from $u$ to $r'$ by $P_a$, from $r'$ to $v$ by $P_b$, and from $v$ to $r$ by $P_v$. We denote the path from $u$ to $u'$ by $P_c$ and the path from $v$ to $v'$ by $P_d$. See the Figure \ref{last figure in lower bound cactus} for an illustration. Note that $G \in \mathcal{G}$.

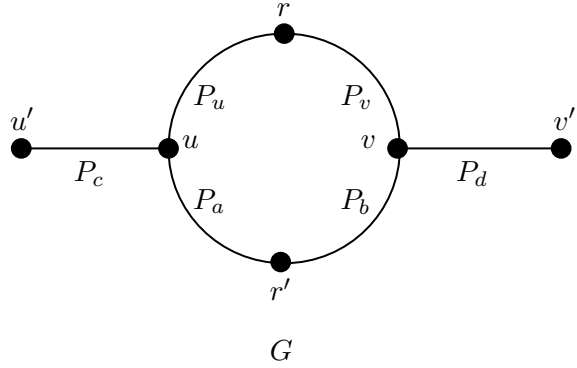
\begin{figure}[H]
    \centering
\tikzset{every picture/.style={line width=0.75pt}} 

\begin{tikzpicture}[x=0.75pt,y=0.75pt,yscale=-1.75,xscale=1.75]

\draw   (84,160) .. controls (84,141.77) and (98.77,127) .. (117,127) .. controls (135.23,127) and (150,141.77) .. (150,160) .. controls (150,178.23) and (135.23,193) .. (117,193) .. controls (98.77,193) and (84,178.23) .. (84,160) -- cycle ;
\draw  [fill={rgb, 255:red, 0; green, 0; blue, 0 }  ,fill opacity=1 ] (114.33,127) .. controls (114.33,125.53) and (115.53,124.33) .. (117,124.33) .. controls (118.47,124.33) and (119.67,125.53) .. (119.67,127) .. controls (119.67,128.47) and (118.47,129.67) .. (117,129.67) .. controls (115.53,129.67) and (114.33,128.47) .. (114.33,127) -- cycle ;
\draw  [fill={rgb, 255:red, 0; green, 0; blue, 0 }  ,fill opacity=1 ] (113.33,192.67) .. controls (113.33,191.19) and (114.53,190) .. (116,190) .. controls (117.47,190) and (118.67,191.19) .. (118.67,192.67) .. controls (118.67,194.14) and (117.47,195.33) .. (116,195.33) .. controls (114.53,195.33) and (113.33,194.14) .. (113.33,192.67) -- cycle ;
\draw  [fill={rgb, 255:red, 0; green, 0; blue, 0 }  ,fill opacity=1 ] (81.33,160) .. controls (81.33,158.53) and (82.53,157.33) .. (84,157.33) .. controls (85.47,157.33) and (86.67,158.53) .. (86.67,160) .. controls (86.67,161.47) and (85.47,162.67) .. (84,162.67) .. controls (82.53,162.67) and (81.33,161.47) .. (81.33,160) -- cycle ;
\draw  [fill={rgb, 255:red, 0; green, 0; blue, 0 }  ,fill opacity=1 ] (146.67,160) .. controls (146.67,158.53) and (147.86,157.33) .. (149.33,157.33) .. controls (150.81,157.33) and (152,158.53) .. (152,160) .. controls (152,161.47) and (150.81,162.67) .. (149.33,162.67) .. controls (147.86,162.67) and (146.67,161.47) .. (146.67,160) -- cycle ;
\draw    (149.33,160) -- (193.33,160) ;
\draw  [fill={rgb, 255:red, 0; green, 0; blue, 0 }  ,fill opacity=1 ] (193.33,160) .. controls (193.33,158.53) and (194.53,157.33) .. (196,157.33) .. controls (197.47,157.33) and (198.67,158.53) .. (198.67,160) .. controls (198.67,161.47) and (197.47,162.67) .. (196,162.67) .. controls (194.53,162.67) and (193.33,161.47) .. (193.33,160) -- cycle ;
\draw    (84,160) -- (42,160) ;
\draw  [fill={rgb, 255:red, 0; green, 0; blue, 0 }  ,fill opacity=1 ] (39.33,160) .. controls (39.33,158.53) and (40.53,157.33) .. (42,157.33) .. controls (43.47,157.33) and (44.67,158.53) .. (44.67,160) .. controls (44.67,161.47) and (43.47,162.67) .. (42,162.67) .. controls (40.53,162.67) and (39.33,161.47) .. (39.33,160) -- cycle ;

\draw (114,117.4) node [anchor=north west][inner sep=0.75pt]    {$r$};
\draw (112,196.4) node [anchor=north west][inner sep=0.75pt]    {$r'$};
\draw (112,214.4) node [anchor=north west][inner sep=0.75pt]    {$G$};
\draw (87,155.4) node [anchor=north west][inner sep=0.75pt]    {$u$};
\draw (138,155.4) node [anchor=north west][inner sep=0.75pt]    {$v$};
\draw (90.17,170.9) node [anchor=north west][inner sep=0.75pt]    {$P_{a}$};
\draw (90.17,140.9) node [anchor=north west][inner sep=0.75pt]    {$P_{u}$};
\draw (56.17,162.9) node [anchor=north west][inner sep=0.75pt]    {$P_{c}$};
\draw (132.17,170.9) node [anchor=north west][inner sep=0.75pt]    {$P_{b}$};
\draw (132.17,140.9) node [anchor=north west][inner sep=0.75pt]    {$P_{v}$};
\draw (165.17,162.9) node [anchor=north west][inner sep=0.75pt]    {$P_{d}$};

\draw (38,147.4) node [anchor=north west][inner sep=0.75pt]    {$u'$};
\draw (193,147.4) node [anchor=north west][inner sep=0.75pt]    {$v'$};

\end{tikzpicture}
    \label{last figure in lower bound cactus}
\end{figure}
\noindent For the sake of contradiction, assume that $w\cdot p < \frac{10}{9}$. Let $k=\frac{w}{2}$, $\mathcal{F}=\mathcal{F}^{w}_{2w}(G,r)$ and $A=\mathcal{F}^{k}_{2w}(G,r)$. Since $K_2 \in \mathcal{G}$ and $\mathcal{G}$ is closed under 1-sum, we can use Theorem~\ref{Theorem: rooted (p,2w)-diameter distribution} and Theorem~\ref{Theorem: radius rooted (p,2w)-diameter distribution} to conclude that there exists $SDD(G,2w,p,r)$ $y=\{y_{F}~|~F\in \mathcal{F}^{w}_{2w}(G,r)\} $ such that, 

$$\sum_{F \in A} y_F = \sum_{F\in \mathcal{F}^{k}_{2w}(r)}y_{F} \geq 1-(w-k) \cdot p=1- \dfrac{w}{2}\cdot p=1-\dfrac{w \cdot p}{2}.$$

  \noindent Suppose that $F\in A$. Since $d(u,r)=k=\frac{w}{2}$ and $d(v,r)=k=\frac{w}{2}$, we have that $F\cap E(P_{v}) \neq \varnothing$ and $F \cap E(P_{u})\neq \varnothing$. The next claim shows that under the assumption that $w \cdot p < \frac{10}{9}$, there exists $F \in A$ which does not pick any edges from $P_a, P_b, P_c,P_d$. This will lead to a contradiction, as $u'$ and $v'$ are $2w$ distance apart, and since $F$ is a $2w$-diameter decomposition, they should have been separated by $F$. This implies that $w \cdot p \geq \frac{10}{9}$ and completes the proof of the theorem. 
    \begin{Claim}
        There exists $F\in A$ such that $F\cap E(P_{a}),F\cap E(P_{b}),F\cap E(P_{c}), F\cap E(P_{d})=\varnothing$.
    \end{Claim}
    \begin{proof}[{Proof of Claim}]
        \renewcommand{\qedsymbol}{$\blacktriangle$}

        Let $P'_{a},P'_{b}$ be the two paths between $r,r'$ containing $P_{a},P_{b}$ respectively. Also, let $P'_{u},P'_{v}$ be the unique shortest paths starting at $r$ and ending at $u',v'$, respectively. Let $A_{a}=\{F\in A~|~F \cap E(P_{a})\neq \varnothing\}$ and $F\in A_{a}$. This implies that $|F\cap E(P^{'}_{a})|\geq 2$. Since $P'_{a}$ is a shortest path with length $w$ from $r$, using Lemma~\ref{lemma:upperbound}, we have, 
        $$\sum_{F \in \mathcal{F}, |F\cap E(P^{'}_{a})| \geq 2}y_{F}\leq (w \cdot p)-1 ~\Rightarrow~\sum_{F\in A_{a}}y_{F}\leq (w \cdot p)-1.$$
         Similarly, we define, 
         $$A_{b}=\{F\in A~|~F \cap E(P_{b})\neq \varnothing \},A_{c}=\{F\in A~|~F \cap E(P_{c})\neq \varnothing\},A_{d}=\{F\in A~|~F \cap E(P_{d})\neq \varnothing\},$$ and doing the same argument for $P'_{b},P'_{u},P'_{v}$, we obtain,
         $$\sum_{F\in A_{b}}y_{F} \leq (w \cdot p)-1;~\sum_{F\in A_{c}}y_{F} \leq (w \cdot p)-1;~\sum_{F\in A_{d}}y_{F}\leq (w \cdot p)-1.$$ Let $A^{*}=A\setminus(A_{a}\cup A_{b}\cup A_{c} \cup A_{d})$. From the discussion above, it follows that,
         $$\sum_{F\in A^{*}}y_{F}=\sum_{F\in A}y_{F}-\sum_{F\in A_{a}\cup A_{b}\cup A_{c}\cup A_{d}} y_{F} \geq \left(1-\dfrac{w \cdot p}{2}\right)-4 \cdot (w \cdot p-1)=5-\dfrac{9 \cdot w \cdot p}{2}>0.$$
        This shows that $A^{*}\neq \varnothing$ and completes the proof of the claim.
    \end{proof}
    
    \end{proof}
In \cref{sec explicit example 20/9}, we use the construction from 
Theorem~\ref{Theorem Cactus 20/9} to exhibit an explicit instance of the minimum multicut problem 
with 
\(\tfrac{\OPT_{IP}}{\OPT_{LP}} \geq \tfrac{20}{9}\).
In the following \cref{sec construction theorem 16/7}, we improve the lower bound to 
\(\tfrac{16}{7}\).

\section{$\tfrac{16}{7}$ Lower Bound for Cactus Graphs}\label{sec construction theorem 16/7}

In this section, we build upon the ideas developed in the previous sections to improve the lower bound to $\frac{16}{7}$. In the proof of Theorem \ref{Theorem Cactus 20/9}, the witness graph $G$ was obtained as the $1$-sum of a cycle of length $2w$ with two simple paths of length $\tfrac{w}{2}$. That argument invoked Theorem \ref{Theorem: radius rooted (p,2w)-diameter distribution} to guarantee the existence of a suitable $SDD(G,2w,p,r)$. We now strengthen this construction by replacing the two paths with carefully chosen cactus “amplifiers.”

The starting point is that for any cactus $G$ and any $r\in V(G)$ and $k\in\{1,\ldots,w\}$, there exists an $SDD(G,2w,p,r)$ such that a $F$ drawn from this distribution satisfies $\rad_F(r)\le k-1$ with probability at least $1-(w-k)p$ (by Theorem \ref{Theorem: radius rooted (p,2w)-diameter distribution}). Our aim is to leverage components for which this success probability is as small as possible (i.e., the “hardest” attachments), and graft them into the cycle gadget.

\begin{Definition}\label{def:R-G-2w-p-k}
Let $\mathcal{G}$ be a graph family containing $K_2$ and closed under taking minors and subdivisions.
Fix $k\in\{1,\ldots,w\}$.
Define
\[
R^{k}_{(\mathcal{G},2w,p)}
\;=\;
\inf_{\substack{G\in\mathcal{G}\\ r\in V(G)}}\;
\max_{\;y\text{ is an }SDD(G,2w,p,r)}
\left\{
\sum_{F\in \mathcal{F}^{k}_{2w}(G,r)} y_F
\right\}.
\]
By Theorem~\ref{Theorem: radius rooted (p,2w)-diameter distribution}, we always have
\(
R^{k}_{(\mathcal{G},2w,p)} \ge 1-(w-k)p.
\)
Moreover, $R^{k}_{(\mathcal{G},2w,p)} \le k\cdot p$:
indeed, for any $G,r$ choose a shortest $r$–$u$ path $Q$ of length $k$.
If $\rad_F(r)\le k-1$ then $F$ must intersect $E(Q)$; hence for any $SDD(G,2w,p,r)$,
\[
\sum_{F\in \mathcal{F}^{k}_{2w}(G,r)} y_F
\;\le\;
\sum_{\substack{F\in \mathcal{F}^{w}_{2w}(G,r)\\ E(Q)\cap F\neq\varnothing}} y_F
\;\le\;
\sum_{e\in E(Q)} \sum_{\substack{F\in \mathcal{F}^{w}_{2w}(G,r)\\ e\in F}} y_F
\;\le\; k\cdot p.
\]

\noindent Therefore,
\begin{align}\label{ upper bound lower bound for R}
    1-(w-k)p \;\le\; R^{k}_{(\mathcal{G},2w,p)} \;\le\; k\cdot p.
\end{align}
This upper bound will be used in the proof of Theorem~\ref{Theorem Cactus 16/7}.
\end{Definition}

\noindent To proceed, we require a structural extension, namely Theorem
\ref{Theorem: generalized radius rooted (p,2w)-diameter distribution}, 
which generalizes Theorem \ref{Theorem: radius rooted (p,2w)-diameter distribution}. 
For this purpose, we introduce the following Definition \ref{def:general-1-sum}:

\begin{Definition}[Generalized $1$-sum]\label{def:general-1-sum}
Let \( G \) be a non-empty graph, and let \( l \) be a natural number.  
Choose vertices \( u_{1}, \ldots, u_{l} \in V(G) \). For each \( i = 1, \ldots, l \), 
let \( G_{i} \) be a non-empty graph with a distinguished vertex \( r_{i} \in V(G_{i}) \). 
Define \( G^{S^{(0)}} := G \). For \( i = 1, \ldots, l \), construct \( G^{S^{(i)}} \) 
by performing the 1-sum of \( G^{S^{(i-1)}} \) and \( G_{i} \) at the vertices \( u_{i} \) and \( r_{i} \). 
After all steps, we obtain the graph \( G^{S^{(l)}} \), which we denote by \( G^{S(L)} \), where  
\[
L = \{(u_{i}, (G_{i}, r_{i})) \mid i = 1, \ldots, l\}.
\]  

\noindent Informally, \( G^{S(L)} \) is the graph obtained by taking the $1$-sum of \( G \) with the graphs 
\( G_{1}, \ldots, G_{l} \), simultaneously identifying \( u_{i} \) with \( r_{i} \) for all \( i = 1, \ldots, l \). 
See Figure \ref{GS(L) Figure} for an illustration:

\begin{figure}[H]
    \centering
    \tikzset{every picture/.style={line width=0.75pt}} 

\begin{tikzpicture}[x=0.75pt,y=0.75pt,yscale=-1,xscale=1]

\draw   (100,131.5) .. controls (100,119.07) and (117.09,109) .. (138.17,109) .. controls (159.25,109) and (176.33,119.07) .. (176.33,131.5) .. controls (176.33,143.93) and (159.25,154) .. (138.17,154) .. controls (117.09,154) and (100,143.93) .. (100,131.5) -- cycle ;
\draw  [fill={rgb, 255:red, 0; green, 0; blue, 0 }  ,fill opacity=1 ] (117.5,150.83) .. controls (117.5,149.82) and (118.32,149) .. (119.33,149) .. controls (120.35,149) and (121.17,149.82) .. (121.17,150.83) .. controls (121.17,151.85) and (120.35,152.67) .. (119.33,152.67) .. controls (118.32,152.67) and (117.5,151.85) .. (117.5,150.83) -- cycle ;
\draw  [fill={rgb, 255:red, 0; green, 0; blue, 0 }  ,fill opacity=1 ] (110.5,125.83) .. controls (110.5,124.82) and (111.32,124) .. (112.33,124) .. controls (113.35,124) and (114.17,124.82) .. (114.17,125.83) .. controls (114.17,126.85) and (113.35,127.67) .. (112.33,127.67) .. controls (111.32,127.67) and (110.5,126.85) .. (110.5,125.83) -- cycle ;
\draw  [fill={rgb, 255:red, 0; green, 0; blue, 0 }  ,fill opacity=1 ] (173.97,132.58) .. controls (173.97,131.57) and (174.79,130.75) .. (175.8,130.75) .. controls (176.81,130.75) and (177.63,131.57) .. (177.63,132.58) .. controls (177.63,133.6) and (176.81,134.42) .. (175.8,134.42) .. controls (174.79,134.42) and (173.97,133.6) .. (173.97,132.58) -- cycle ;
\draw   (176.33,131.5) .. controls (187.97,127.15) and (203.39,139.63) .. (210.77,159.37) .. controls (218.16,179.11) and (214.71,198.65) .. (203.07,203) .. controls (191.43,207.35) and (176.01,194.87) .. (168.62,175.13) .. controls (161.24,155.39) and (164.69,135.85) .. (176.33,131.5) -- cycle ;
\draw   (111.86,130.31) .. controls (111.08,140.89) and (97.92,148.53) .. (82.46,147.38) .. controls (67.01,146.23) and (55.12,136.73) .. (55.91,126.15) .. controls (56.7,115.57) and (69.86,107.93) .. (85.31,109.08) .. controls (100.76,110.23) and (112.65,119.73) .. (111.86,130.31) -- cycle ;
\draw   (121.17,150.83) .. controls (138.7,154.89) and (149.94,170.99) .. (146.29,186.8) .. controls (142.63,202.6) and (125.45,212.12) .. (107.92,208.07) .. controls (90.39,204.01) and (79.14,187.91) .. (82.8,172.1) .. controls (86.46,156.3) and (103.64,146.78) .. (121.17,150.83) -- cycle ;

\draw (118,118.4) node [anchor=north west][inner sep=0.75pt]  [font=\scriptsize]  {$u_{1} =r_{1}$};
\draw (179,118.4) node [anchor=north west][inner sep=0.75pt]  [font=\scriptsize]  {$u_{3} =r_{3}$};
\draw (95,162.4) node [anchor=north west][inner sep=0.75pt]  [font=\scriptsize]  {$u_{2} =r_{2}$};
\draw (69,120.4) node [anchor=north west][inner sep=0.75pt]    {$G_{1}$};
\draw (101,185.4) node [anchor=north west][inner sep=0.75pt]    {$G_{2}$};
\draw (185,165.4) node [anchor=north west][inner sep=0.75pt]    {$G_{3}$};
\draw (130,131.4) node [anchor=north west][inner sep=0.75pt]    {$G$};

\end{tikzpicture}
    \caption{The construction of $G^{S(L)}$ with 
    $L=\{(u_{1},(G_{1},r_{1})),(u_{2},(G_{2},r_{2})),(u_{3},(G_{3},r_{3}))\}$.}
    \label{GS(L) Figure}
\end{figure}
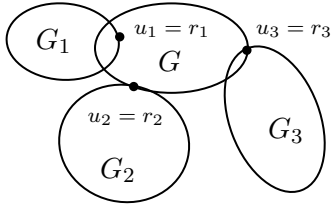
\end{Definition}

\begin{Theorem}\label{Theorem: generalized radius rooted (p,2w)-diameter distribution}
Suppose that $\mathcal{G}$ is closed under the 1-sum operation, contains $K_2$, and is strongly $SDD(2w,p)$-acceptable.  
Let $G \in \mathcal{G}$, $r \in V(G)$, and $k \in \{1,\ldots,w\}$. Let $l$ be a natural number, and let $u_{1}, \ldots, u_{l} \in V(G)$ be arbitrary vertices. For each $i = 1, \ldots, l$, let $G_{i} \in \mathcal{G}$ with a distinguished vertex $r_{i} \in V(G_{i})$. Define
\[
L = \{ (u_{i}, (G_{i}, r_{i})) \mid i = 1, \ldots, l \}.
\]
Then there exists an $SDD(G^{S(L)},2w,p,r)$ distribution 
\(
y = \{ y_{F} \mid F \in \mathcal{F}^{w}_{2w}(G^{S(L)}) \}
\)
for $G^{S(L)}$ such that
\[
\sum_{F \in \mathcal{F}^{k}_{2w}(G^{S(L)}, r)} y_{F} \;\; \geq \;\; R^{k}_{(\mathcal{G},2w,p)}\geq 1-(w-k)p,
\]
and moreover, for each $i = 1, \ldots, l$, the projection $y(G_{i})$ is an $SDD(G_{i},2w,p,r_{i})$ for $G_{i}$.
\end{Theorem}

\begin{proof}

Let \(\mathcal{F}^{w}_{2w}=\mathcal{F}^{w}_{2w}(G^{S(L)},r)\) and 
\(\mathcal{F}^{k}_{2w}=\mathcal{F}^{k}_{2w}(G^{S(L)},r)\).
For each \(i=1,\ldots,l\), define
\[
\mathcal{F}_{i}
=\Bigl\{F\in \mathcal{F}^{w}_{2w}\ \Big|\ F\cap E(G_{i})\in \mathcal{F}^{w}_{2w}(G_{i},r_{i})\Bigr\},
\qquad
\mathcal{F}=\bigcap_{i=1}^{l}\mathcal{F}_{i}.
\]
It suffices to show that the following LP is feasible and has optimal value \(0\):
\begin{equation}\label{LP:3}
\begin{aligned}
\min\ &\sum_{F\in \mathcal{F}^{w}_{2w}\setminus \mathcal{F}} y_{F} \\
\text{s.t.}\quad
&\sum_{F\in \mathcal{F}^{k}_{2w}} y_{F}\ \ge\ R^{k}_{(\mathcal{G},2w,p)},\\
&\sum_{\substack{F\in \mathcal{F}^{w}_{2w}\\ e\in F}} y_{F}\ \le\ p 
\qquad \forall\, e\in E\bigl(G^{S(L)}\bigr),\\
&\sum_{F\in \mathcal{F}^{w}_{2w}} y_{F}=1,\qquad
y_{F}\ge 0\ \ \forall\,F\in \mathcal{F}^{w}_{2w}.
\end{aligned}
\end{equation}

By Theorem~\ref{Theorem: radius rooted (p,2w)-diameter distribution} and by the definition of 
\(R^{k}_{(\mathcal{G},2w,p)}\) (see Definition~\ref{def:R-G-2w-p-k}), LP~\eqref{LP:3} is feasible.
Suppose, for contradiction, that the optimal value of \eqref{LP:3} is \(z>0\).
Let \(m>\tfrac{l^{2}}{z}\).
For each \(i=1,\ldots,l\), let \(G_{i}^{1},\ldots,G_{i}^{m}\) be \(m\) disjoint copies of \(G_{i}\),
and let \(r_{i}^{j}\in V(G_{i}^{j})\) be the vertex corresponding to \(r_{i}\).
Form \(G'_{i}\) by taking the \(1\)-sum of \(G_{i}^{1},\ldots,G_{i}^{m}\) at 
\(r_{i}^{1},\ldots,r_{i}^{m}\), and denote the identified vertex again by \(r_{i}\).
Define
\[
L'=\{(u_{i},(G'_{i},r_{i}))\mid i=1,\ldots,l\}.
\]
Since \(\mathcal{G}\) is closed under the \(1\)-sum operation, we have \(G^{S(L')}\in \mathcal{G}\).
See Figure~\ref{the figure of GSL' and GSL toether} for an illustration:

\begin{figure}[H]
    \centering
    \begin{minipage}{0.48\textwidth}
        \centering
        \tikzset{every picture/.style={line width=0.75pt}} 

\begin{tikzpicture}[x=0.75pt,y=0.75pt,yscale=-1,xscale=1]

\draw   (100,131.5) .. controls (100,119.07) and (117.09,109) .. (138.17,109) .. controls (159.25,109) and (176.33,119.07) .. (176.33,131.5) .. controls (176.33,143.93) and (159.25,154) .. (138.17,154) .. controls (117.09,154) and (100,143.93) .. (100,131.5) -- cycle ;
\draw  [fill={rgb, 255:red, 0; green, 0; blue, 0 }  ,fill opacity=1 ] (117.5,150.83) .. controls (117.5,149.82) and (118.32,149) .. (119.33,149) .. controls (120.35,149) and (121.17,149.82) .. (121.17,150.83) .. controls (121.17,151.85) and (120.35,152.67) .. (119.33,152.67) .. controls (118.32,152.67) and (117.5,151.85) .. (117.5,150.83) -- cycle ;
\draw  [fill={rgb, 255:red, 0; green, 0; blue, 0 }  ,fill opacity=1 ] (110.5,125.83) .. controls (110.5,124.82) and (111.32,124) .. (112.33,124) .. controls (113.35,124) and (114.17,124.82) .. (114.17,125.83) .. controls (114.17,126.85) and (113.35,127.67) .. (112.33,127.67) .. controls (111.32,127.67) and (110.5,126.85) .. (110.5,125.83) -- cycle ;
\draw  [fill={rgb, 255:red, 0; green, 0; blue, 0 }  ,fill opacity=1 ] (173.97,132.58) .. controls (173.97,131.57) and (174.79,130.75) .. (175.8,130.75) .. controls (176.81,130.75) and (177.63,131.57) .. (177.63,132.58) .. controls (177.63,133.6) and (176.81,134.42) .. (175.8,134.42) .. controls (174.79,134.42) and (173.97,133.6) .. (173.97,132.58) -- cycle ;
\draw   (176.33,131.5) .. controls (187.97,127.15) and (203.39,139.63) .. (210.77,159.37) .. controls (218.16,179.11) and (214.71,198.65) .. (203.07,203) .. controls (191.43,207.35) and (176.01,194.87) .. (168.62,175.13) .. controls (161.24,155.39) and (164.69,135.85) .. (176.33,131.5) -- cycle ;
\draw   (111.86,130.31) .. controls (111.08,140.89) and (97.92,148.53) .. (82.46,147.38) .. controls (67.01,146.23) and (55.12,136.73) .. (55.91,126.15) .. controls (56.7,115.57) and (69.86,107.93) .. (85.31,109.08) .. controls (100.76,110.23) and (112.65,119.73) .. (111.86,130.31) -- cycle ;
\draw   (121.17,150.83) .. controls (138.7,154.89) and (149.94,170.99) .. (146.29,186.8) .. controls (142.63,202.6) and (125.45,212.12) .. (107.92,208.07) .. controls (90.39,204.01) and (79.14,187.91) .. (82.8,172.1) .. controls (86.46,156.3) and (103.64,146.78) .. (121.17,150.83) -- cycle ;

\draw (118,118.4) node [anchor=north west][inner sep=0.75pt]  [font=\scriptsize]  {$u_{1} =r_{1}$};
\draw (179,118.4) node [anchor=north west][inner sep=0.75pt]  [font=\scriptsize]  {$u_{3} =r_{3}$};
\draw (95,162.4) node [anchor=north west][inner sep=0.75pt]  [font=\scriptsize]  {$u_{2} =r_{2}$};
\draw (69,120.4) node [anchor=north west][inner sep=0.75pt]    {$G_{1}$};
\draw (101,185.4) node [anchor=north west][inner sep=0.75pt]    {$G_{2}$};
\draw (185,165.4) node [anchor=north west][inner sep=0.75pt]    {$G_{3}$};
\draw (130,131.4) node [anchor=north west][inner sep=0.75pt]    {$G$};

\end{tikzpicture}
    \end{minipage}\hfill
    \begin{minipage}{0.48\textwidth}
        \centering
        \tikzset{every picture/.style={line width=0.75pt}} 

\begin{tikzpicture}[x=0.75pt,y=0.75pt,yscale=-1,xscale=1]

\draw   (120,140.5) .. controls (120,128.07) and (137.09,118) .. (158.17,118) .. controls (179.25,118) and (196.33,128.07) .. (196.33,140.5) .. controls (196.33,152.93) and (179.25,163) .. (158.17,163) .. controls (137.09,163) and (120,152.93) .. (120,140.5) -- cycle ;
\draw  [fill={rgb, 255:red, 0; green, 0; blue, 0 }  ,fill opacity=1 ] (137.5,159.83) .. controls (137.5,158.82) and (138.32,158) .. (139.33,158) .. controls (140.35,158) and (141.17,158.82) .. (141.17,159.83) .. controls (141.17,160.85) and (140.35,161.67) .. (139.33,161.67) .. controls (138.32,161.67) and (137.5,160.85) .. (137.5,159.83) -- cycle ;
\draw  [fill={rgb, 255:red, 0; green, 0; blue, 0 }  ,fill opacity=1 ] (130.5,134.83) .. controls (130.5,133.82) and (131.32,133) .. (132.33,133) .. controls (133.35,133) and (134.17,133.82) .. (134.17,134.83) .. controls (134.17,135.85) and (133.35,136.67) .. (132.33,136.67) .. controls (131.32,136.67) and (130.5,135.85) .. (130.5,134.83) -- cycle ;
\draw  [fill={rgb, 255:red, 0; green, 0; blue, 0 }  ,fill opacity=1 ] (193.97,141.58) .. controls (193.97,140.57) and (194.79,139.75) .. (195.8,139.75) .. controls (196.81,139.75) and (197.63,140.57) .. (197.63,141.58) .. controls (197.63,142.6) and (196.81,143.42) .. (195.8,143.42) .. controls (194.79,143.42) and (193.97,142.6) .. (193.97,141.58) -- cycle ;
\draw   (195.8,141.58) .. controls (207.44,137.23) and (222.86,149.71) .. (230.24,169.45) .. controls (237.62,189.2) and (234.17,208.73) .. (222.53,213.08) .. controls (210.89,217.43) and (195.47,204.96) .. (188.09,185.21) .. controls (180.71,165.47) and (184.16,145.94) .. (195.8,141.58) -- cycle ;
\draw   (131.86,139.31) .. controls (131.08,149.89) and (117.92,157.53) .. (102.46,156.38) .. controls (87.01,155.23) and (75.12,145.73) .. (75.91,135.15) .. controls (76.7,124.57) and (89.86,116.93) .. (105.31,118.08) .. controls (120.76,119.23) and (132.65,128.73) .. (131.86,139.31) -- cycle ;
\draw   (141.17,159.83) .. controls (158.7,163.89) and (169.94,179.99) .. (166.29,195.8) .. controls (162.63,211.6) and (145.45,221.12) .. (127.92,217.07) .. controls (110.39,213.01) and (99.14,196.91) .. (102.8,181.1) .. controls (106.46,165.3) and (123.64,155.78) .. (141.17,159.83) -- cycle ;
\draw   (201,137.97) .. controls (211.8,131.82) and (229.01,141.67) .. (239.45,159.98) .. controls (249.89,178.29) and (249.61,198.13) .. (238.81,204.28) .. controls (228.02,210.44) and (210.8,200.58) .. (200.36,182.27) .. controls (189.92,163.96) and (190.21,144.13) .. (201,137.97) -- cycle ;
\draw   (195.8,141.58) .. controls (201.34,130.46) and (221.12,129.05) .. (239.99,138.45) .. controls (258.87,147.84) and (269.67,164.47) .. (264.14,175.6) .. controls (258.6,186.72) and (238.81,188.12) .. (219.94,178.73) .. controls (201.07,169.34) and (190.26,152.71) .. (195.8,141.58) -- cycle ;
\draw   (141.17,159.83) .. controls (159.16,160.26) and (173.42,173.76) .. (173.04,189.98) .. controls (172.65,206.19) and (157.75,218.99) .. (139.76,218.56) .. controls (121.77,218.13) and (107.5,204.64) .. (107.89,188.42) .. controls (108.28,172.2) and (123.18,159.4) .. (141.17,159.83) -- cycle ;
\draw   (141.17,159.83) .. controls (159.08,158.09) and (174.87,169.77) .. (176.44,185.92) .. controls (178.01,202.06) and (164.76,216.56) .. (146.85,218.3) .. controls (128.94,220.04) and (113.15,208.37) .. (111.58,192.22) .. controls (110.01,176.07) and (123.26,161.57) .. (141.17,159.83) -- cycle ;
\draw   (132.33,133) .. controls (129.75,143.29) and (115.48,148.57) .. (100.45,144.8) .. controls (85.42,141.03) and (75.33,129.64) .. (77.91,119.35) .. controls (80.49,109.06) and (94.77,103.78) .. (109.79,107.55) .. controls (124.82,111.32) and (134.91,122.71) .. (132.33,133) -- cycle ;
\draw   (132.33,133) .. controls (124.48,140.12) and (109.67,136.59) .. (99.26,125.12) .. controls (88.85,113.64) and (86.79,98.56) .. (94.64,91.43) .. controls (102.5,84.31) and (117.31,87.84) .. (127.72,99.32) .. controls (138.12,110.8) and (140.19,125.88) .. (132.33,133) -- cycle ;

\draw (138,127.4) node [anchor=north west][inner sep=0.75pt]  [font=\scriptsize]  {$u_{1} =r_{1}$};
\draw (195,125.4) node [anchor=north west][inner sep=0.75pt]  [font=\scriptsize]  {$u_{3} =r_{3}$};
\draw (123,171.4) node [anchor=north west][inner sep=0.75pt]  [font=\scriptsize]  {$u_{2} =r_{2}$};
\draw (49,101.4) node [anchor=north west][inner sep=0.75pt]    {$G'_{1}$};
\draw (118,223.4) node [anchor=north west][inner sep=0.75pt]    {$G'_{2}$};
\draw (251,195.4) node [anchor=north west][inner sep=0.75pt]    {$G'_{3}$};
\draw (150,138.4) node [anchor=north west][inner sep=0.75pt]    {$G$};

\end{tikzpicture}
    \end{minipage}
    \caption{The left graph is \(G^{S(L)}\) and the right graph is \(G^{S(L')}\).}
    \label{the figure of GSL' and GSL toether}
\end{figure}
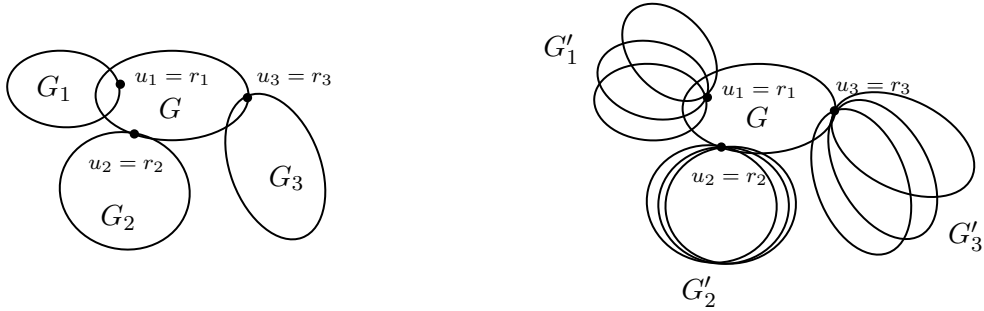

\noindent By the assumption that $\mathcal{G}$ is strongly $SDD(2w,p)$-acceptable, and using 
Theorem~\ref{Theorem: radius rooted (p,2w)-diameter distribution} together with 
Definition~\ref{def:R-G-2w-p-k}, there exists an $SDD(G^{S(L')},2w,p,r)$ distribution
\[
x=\{x_{F'} \mid F'\in \mathcal{F}^{w}_{2w}(G^{S(L')})\},
\]
with the property that
\[
\sum_{F'\in\mathcal{F}^{k}_{2w}(G^{S(L')},r)} x_{F'} \;\;\geq\;\; R^{k}_{(\mathcal{G},2w,p)}.
\]

\noindent For each $j=1,\ldots,m$, let $H_{j}$ be the induced subgraph of $G^{S(L')}$ on the vertex set 
\[
V(H_{j}) \;=\; V(G)\;\cup\;\Bigl(\bigcup_{i=1}^{l} V(G_{i}^{j})\Bigr).
\]
Observe that each $H_{j}$ is an identical copy of $G^{S(L)}$. Moreover, by construction of $G^{S(L')}$ and the subgraphs $H_{j}$, the following two properties hold for every fixed $j\in\{1,\ldots,m\}$:
\begin{enumerate}
    \item For each vertex $v\in V(H_{j})$, we have 
    \[
    d_{H_{j}}(r,v) = d_{G^{S(L')}}(r,v).
    \]
    \item For each $i\in\{1,\ldots,l\}$ and each vertex $v\in V(G_{i}^{j})$, we have 
    \[
    d_{G_{i}^{j}}(v,r_{i}) = d_{H_{j}}(v,r_{i}) = d_{G^{S(L')}}(v,r_{i}).
    \]
\end{enumerate}

\begin{Claim}
Fix $j\in\{1,\ldots,m\}$. The projection $x(H_{j})$ is a feasible solution to LP~\eqref{LP:3}.
\end{Claim}

\begin{proof}[Proof of Claim]
\renewcommand{\qedsymbol}{$\blacktriangle$}
Since $x$ is an $SDD(G^{S(L')},2w,p)$, Lemma~\ref{Lemma: projection is pload} implies that 
$x(H_{j})$ is an $SDD(H_{j},2w,p)$. Note that $r \in V(G)$, and hence $r \in V(H_{j})$. 
Moreover, since 
\begin{align}\label{Strong SDD for GSL'}
    \sum_{F' \in \mathcal{F}^{w}_{2w}(G^{S(L')},r)} x_{F'} = 1,
\end{align}
we must also have
\begin{align}\label{Strong SDD for the projections}
    \sum_{F'' \in \mathcal{F}^{w}_{2w}(H_{j},r)} x(H_{j})_{F''} = 1.
\end{align}

\noindent To see why \eqref{Strong SDD for the projections} holds, suppose not. Then there exists 
$F'' \in \mathcal{F}_{2w}(H_{j}) \setminus \mathcal{F}^{w}_{2w}(H_{j},r)$ with 
$x(H_{j})_{F''} > 0$. By definition of projection, this implies that there exists 
$F' \in \mathcal{F}_{2w}(G^{S(L')})$ with $x_{F'} > 0$ and $F' \cap E(H_{j}) = F''$. 
Now, $F'' \notin \mathcal{F}^{w}_{2w}(H_{j},r)$ means that in $H_{j} \setminus F''$ 
there exists some vertex $v \in V(H_{j})$ with $d_{H_{j}}(r,v) \geq w$. As mentioned earlier
\[
d_{G^{S(L')}}(r,v) =d_{H_{j}}(r,v) \geq w.
\]
Thus, $r$ and $v$ lie in the same connected component of $G^{S(L')} \setminus F'$, 
and so $F' \notin \mathcal{F}^{w}_{2w}(G^{S(L')},r)$. 
This contradicts \eqref{Strong SDD for GSL'}, since $x_{F'} > 0$. 
Hence, \eqref{Strong SDD for the projections} must hold, and therefore 
$x(H_{j})$ is an $SDD(H_{j},2w,p,r)$. It remains to show that
\[
\sum_{F'' \in \mathcal{F}^{k}_{2w}(H_{j},r)} x(H_{j})_{F''} \;\;\geq\;\; R^{k}_{(\mathcal{G},2w,p)}.
\]

\noindent Indeed, for any $F' \in \mathcal{F}^{w}_{2w}(G^{S(L')},r)$ we have 
$F'' = F' \cap E(H_{j}) \in \mathcal{F}^{k}_{2w}(H_{j},r)$ whenever 
$F' \in \mathcal{F}^{k}_{2w}(G^{S(L')},r)$. Therefore,
\begin{align*}
\sum_{F'' \in \mathcal{F}^{k}_{2w}(H_{j},r)} x(H_{j})_{F''}
&= \sum_{F'' \in \mathcal{F}^{k}_{2w}(H_{j},r)} 
   \sum_{\substack{F' \in \mathcal{F}_{2w}(G^{S(L')}) \\ F'' = F' \cap E(H_{j})}} x_{F'} \\
&\geq \sum_{F' \in \mathcal{F}^{k}_{2w}(G^{S(L')},r)} x_{F'} \\
&\geq R^{k}_{(\mathcal{G},2w,p)},
\end{align*}
where the last inequality follows from the defining property of $x$.
\end{proof}

For each $j \in \{1,\ldots,m\}$ and $i \in \{1,\ldots,l\}$, let $\mathcal{F}_{i}^{j}$ denote the copy of $\mathcal{F}_{i}$ inside $H_{j}$. 
Also, let $\mathcal{F}^{j} = \bigcap_{i=1}^{l} \mathcal{F}_{i}^{j}$ be the copy of $\mathcal{F}$ in $H_{j}$. 
Since $x(H_{j})$ is a feasible solution to LP~\eqref{LP:3}, we obtain
\begin{align*}
    \sum_{F_{j}\in \mathcal{F}^{w}_{2w}(H_{j}) \setminus \mathcal{F}^{j}} x(H_{j})_{F_{j}} \;\;\geq\; z.
\end{align*}

\noindent Hence, for each $j \in \{1,\ldots,m\}$ there exists an index $i_{j} \in \{1,\ldots,l\}$ such that
\begin{align}\label{bad 2wdiameteer for Hj in Gi}
    \sum_{F_{j}\in \mathcal{F}^{w}_{2w}(H_{j}) \setminus \mathcal{F}_{i_{j}}^{j}} x(H_{j})_{F_{j}} \;\;\geq\; \frac{z}{l}.
\end{align}

\noindent The indices $i_{j}$ play a crucial role. 
By the Pigeonhole Principle, for at least 
\[
t = \Big\lceil \frac{m}{l} \Big\rceil
\]
values of $j$, the same index $i_{j}$ is chosen. Without loss of generality, assume $i_{1} = \cdots = i_{t} = 1$. 
Summing \eqref{bad 2wdiameteer for Hj in Gi} over $j=1,\ldots,t$ gives
\begin{align}\label{bad 2wdiameter for the same i for many j's}
    \sum_{j=1}^{t} \sum_{F_{j}\in \mathcal{F}^{w}_{2w}(H_{j}) \setminus \mathcal{F}_{1}^{j}} x(H_{j})_{F_{j}}
    \;\;\geq\; \sum_{j=1}^{t} \frac{z}{l} = \frac{tz}{l} 
    \;\;\geq\; \frac{m}{l} \cdot \frac{z}{l} \;>\; 1.
\end{align}

\noindent The left-hand side of \eqref{bad 2wdiameter for the same i for many j's} can be rewritten as
\begin{align}\label{bad 2wdiameter expansion}
    \sum_{j=1}^{t} \sum_{F_{j}\in \mathcal{F}^{w}_{2w}(H_{j}) \setminus \mathcal{F}_{1}^{j}}
        \;\;\sum_{\substack{F'\in \mathcal{F}_{2w}(G^{S(L')}) \\ F_{j}=F'\cap E(H_{j})}} x_{F'}.
\end{align}
Since
\[
\sum_{F' \in \mathcal{F}_{2w}(G^{S(L')})} x_{F'} 
= \sum_{F' \in \mathcal{F}^{w}_{2w}(G^{S(L')})} x_{F'} 
= 1,
\]
it follows that some $F' \in \mathcal{F}^{w}_{2w}(G^{S(L')})$ with $x_{F'}>0$ must appear at least twice in the expansion \eqref{bad 2wdiameter expansion}. 

\noindent Concretely, there exist distinct indices $j_{1},j_{2} \in \{1,\ldots,t\}$ such that 
\[
F_{j_{1}} = F'\cap E(H_{j_{1}}) \in \mathcal{F}^{w}_{2w}(H_{j_{1}}) \setminus \mathcal{F}_{1}^{j_{1}}, 
\quad 
F_{j_{2}} = F'\cap E(H_{j_{2}}) \in \mathcal{F}^{w}_{2w}(H_{j_{2}}) \setminus \mathcal{F}_{1}^{j_{2}}.
\]

\noindent The first condition implies that in $H_{j_{1}}\setminus F_{j_{1}}$ there exists a vertex $v_{1}\in V(G^{j_{1}}_{1})$ such that 
$d_{H_{j_{1}}}(r_{1},v_{1}) \geq w$. As mentioned earlier, we also have 
$d_{G^{S(L')}}(r_{1},v_{1}) \geq w$. 
Similarly, from the second condition there exists $v_{2}\in V(G^{j_{2}}_{1})$ with 
$d_{G^{S(L')}}(r_{1},v_{2}) \geq w$. Therefore, in $G^{S(L')}\setminus F'$ the vertices $v_{1}, v_{2}, r_{1}$ all lie in the same connected component. Moreover, since $r_{1}$ is a cut vertex separating $v_{1}$ and $v_{2}$, we have
\[
d_{G^{S(L')}}(v_{1},v_{2}) 
= d_{G^{S(L')}}(r_{1},v_{1}) + d_{G^{S(L')}}(r_{1},v_{2}) 
\;\;\geq\; 2w.
\]
This contradicts the assumption that $F' \in \mathcal{F}^{w}_{2w}(G^{S(L')})$. 

\noindent Hence, our initial assumption that the optimal value $z$ of LP~\eqref{LP:3} is positive must be false. Therefore, the optimal value is $z=0$, completing the proof.

        \end{proof}

We work in the same regime as before, where \(w\) is an even integer. For the remainder, set \(k=\tfrac{w}{2}\) and fix a small constant \(\epsilon>0\).
Choose a cactus \(H'\in\mathcal{G}\) and a distinguished vertex \(r'\in V(H')\) such that
\begin{align}\label{choice of H'}
  \max_{\;y\text{ is an }SDD(H',2w,p,r')}
  \left\{
    \sum_{F\in \mathcal{F}^{k}_{2w}(H',r')} y_F
  \right\}
  \;<\;
  R^{k}_{(\mathcal{G},2w,p)} + \epsilon.
\end{align}

\noindent
Now, considering this extremal pair \((H',r')\), we replace the two paths used in
Theorem~\ref{Theorem Cactus 20/9} with copies of \(H'\) attached at the cycle vertices \(u_1\) and \(u_2\)
(via the distinguished vertices \(r'_1,r'_2\) of the copies).
In other words, we work with the graph obtained by taking the \(1\)-sum of the cycle and two disjoint
copies of \(H'\) at \(u_1\) and \(u_2\).

\begin{Theorem}\label{Theorem Cactus 16/7}
Let \(\mathcal{G}\) be the family of cactus graphs. If \(\mathcal{G}\) is strongly \(SDD(2w,p)\)-acceptable,
then \(w \cdot p \geq \tfrac{8}{7}-\tfrac{4\epsilon}{7}\).
\end{Theorem}

\begin{proof}
Let \(H\) be a cycle of length \(2w\). Let \(r, u_{1}, \tilde{r}, u_{2}\) be four vertices of the cycle
in anticlockwise order, such that
\(d(r,u_{1}) = d(u_{1},\tilde{r}) = d(\tilde{r},u_{2}) = d(u_{2},r)=\tfrac{w}{2}\).
Let \(H'_{1},H'_{2}\) be two copies of \(H'\) with corresponding distinguished vertices \(r'_{1},r'_{2}\).
Set
\[
G \;=\; H^{S\bigl(\{(u_{1},(H'_{1},r'_{1})),\ (u_{2},(H'_{2},r'_{2}))\}\bigr)}.
\]
Then \(G \in \mathcal{G}\).
Denote by \(P_{u_{1}}\) the path of length \(\tfrac{w}{2}\) from \(r\) to \(u_{1}\),
by \(P_{a}\) the path from \(u_{1}\) to \(\tilde{r}\),
by \(P_{b}\) the path from \(\tilde{r}\) to \(u_{2}\),
and by \(P_{u_{2}}\) the path from \(u_{2}\) to \(r\).
See Figure~\ref{last figure in 16/7 cactus} for an illustration:

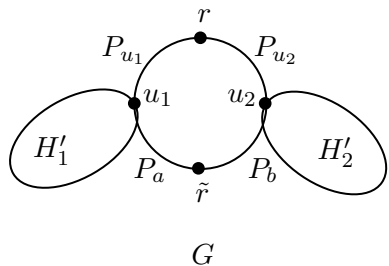
\begin{figure}[H]
    \centering
    \tikzset{every picture/.style={line width=0.75pt}} 

\begin{tikzpicture}[x=0.75pt,y=0.75pt,yscale=-1,xscale=1]

\draw   (217,125) .. controls (217,106.77) and (231.77,92) .. (250,92) .. controls (268.23,92) and (283,106.77) .. (283,125) .. controls (283,143.23) and (268.23,158) .. (250,158) .. controls (231.77,158) and (217,143.23) .. (217,125) -- cycle ;
\draw  [fill={rgb, 255:red, 0; green, 0; blue, 0 }  ,fill opacity=1 ] (247.33,92) .. controls (247.33,90.53) and (248.53,89.33) .. (250,89.33) .. controls (251.47,89.33) and (252.67,90.53) .. (252.67,92) .. controls (252.67,93.47) and (251.47,94.67) .. (250,94.67) .. controls (248.53,94.67) and (247.33,93.47) .. (247.33,92) -- cycle ;
\draw  [fill={rgb, 255:red, 0; green, 0; blue, 0 }  ,fill opacity=1 ] (246.33,157.67) .. controls (246.33,156.19) and (247.53,155) .. (249,155) .. controls (250.47,155) and (251.67,156.19) .. (251.67,157.67) .. controls (251.67,159.14) and (250.47,160.33) .. (249,160.33) .. controls (247.53,160.33) and (246.33,159.14) .. (246.33,157.67) -- cycle ;
\draw  [fill={rgb, 255:red, 0; green, 0; blue, 0 }  ,fill opacity=1 ] (214.33,125) .. controls (214.33,123.53) and (215.53,122.33) .. (217,122.33) .. controls (218.47,122.33) and (219.67,123.53) .. (219.67,125) .. controls (219.67,126.47) and (218.47,127.67) .. (217,127.67) .. controls (215.53,127.67) and (214.33,126.47) .. (214.33,125) -- cycle ;
\draw  [fill={rgb, 255:red, 0; green, 0; blue, 0 }  ,fill opacity=1 ] (279.67,125) .. controls (279.67,123.53) and (280.86,122.33) .. (282.33,122.33) .. controls (283.81,122.33) and (285,123.53) .. (285,125) .. controls (285,126.47) and (283.81,127.67) .. (282.33,127.67) .. controls (280.86,127.67) and (279.67,126.47) .. (279.67,125) -- cycle ;
\draw   (156.64,160.45) .. controls (151.05,150.93) and (160.02,135.27) .. (176.69,125.48) .. controls (193.36,115.69) and (211.41,115.48) .. (217,125) .. controls (222.59,134.52) and (213.62,150.18) .. (196.95,159.97) .. controls (180.28,169.76) and (162.24,169.98) .. (156.64,160.45) -- cycle ;
\draw   (340.56,164.83) .. controls (334.28,173.91) and (316.3,172.36) .. (300.4,161.36) .. controls (284.51,150.36) and (276.72,134.08) .. (283,125) .. controls (289.28,115.92) and (307.27,117.47) .. (323.16,128.47) .. controls (339.06,139.47) and (346.85,155.75) .. (340.56,164.83) -- cycle ;

\draw (247,76.4) node [anchor=north west][inner sep=0.75pt]    {$r$};
\draw (246,163.4) node [anchor=north west][inner sep=0.75pt]    {$\Tilde{r} $};
\draw (220,116.4) node [anchor=north west][inner sep=0.75pt]    {$u_{1}$};
\draw (262,116.4) node [anchor=north west][inner sep=0.75pt]    {$u_{2}$};
\draw (272.17,150.9) node [anchor=north west][inner sep=0.75pt]    {$P_{b}$};
\draw (215.17,150.9) node [anchor=north west][inner sep=0.75pt]    {$P_{a}$};
\draw (199.17,89.9) node [anchor=north west][inner sep=0.75pt]    {$P_{u_{1}}$};
\draw (275.17,89.9) node [anchor=north west][inner sep=0.75pt]    {$P_{u_{2}}$};
\draw (165,139.4) node [anchor=north west][inner sep=0.75pt]    {$H'_{1}$};
\draw (307,140.4) node [anchor=north west][inner sep=0.75pt]    {$H'_{2}$};
\draw (244,194.4) node [anchor=north west][inner sep=0.75pt]    {$G$};

\end{tikzpicture}
    \label{last figure in 16/7 cactus}
\end{figure}

\noindent
Recall that \(k=\tfrac{w}{2}\).
By Theorem~\ref{Theorem: generalized radius rooted (p,2w)-diameter distribution}, there exists an
\(SDD(G,2w,p,r)\) distribution
\(
y = \{ y_{F} \mid F \in \mathcal{F}^{w}_{2w}(G,r) \}
\)
such that
\begin{align}\label{A before def 16/7}
    \sum_{F \in \mathcal{F}^{k}_{2w}(G, r)} y_{F}
    \;\; \geq \;\; R^{k}_{(\mathcal{G},2w,p)}
    \;\; \geq \;\; 1-(w-k)p
    \;=\; 1-\frac{w\cdot p}{2},
\end{align}
and moreover, for each \(i = 1,2\), the projection \(y(H'_{i})\) is an \(SDD(H'_{i},2w,p,r'_{i})\).

\noindent
For the sake of contradiction, assume that \(w\cdot p < \tfrac{8}{7}-\tfrac{4\epsilon}{7}\).
Let \(\mathcal{F}=\mathcal{F}^{w}_{2w}(G,r)\) and \(A=\mathcal{F}^{k}_{2w}(G,r)\).
For \(i=1,2\), let \(\mathcal{F}_{i}=\mathcal{F}^{w}_{2w}(H'_{i},r'_{i})\) and
\(\mathcal{F}^{k}_{i}=\mathcal{F}^{k}_{2w}(H'_{i},r'_{i})\).
From \eqref{A before def 16/7} we have
\begin{align}\label{A 16/7}
    \sum_{F\in A} y_{F} \;\geq\; R^{k}_{(\mathcal{G},2w,p)}.
\end{align}

\noindent
In the following we are going to show that under the assumption \(w \cdot p < \tfrac{8}{7}-\tfrac{4\epsilon}{7}\)
there exists \(F \in A\) with $y_{F}>0$ which (i) does not pick any edges from \(P_a\) or \(P_b\),
and (ii) fails the radius-\(k\) condition inside both attachments:
\(F\cap E(H'_{i}) \notin \mathcal{F}_{i}^{k}\) for \(i=1,2\).
This leads to a contradiction, since \(d(u_{1},u_{2})=w\) and for each \(i\) there exists
\(v_{i}\in V(H'_{i})\) with \(d(u_{i},v_{i})\geq k=\tfrac{w}{2}\).
Thus \(d(v_{1},v_{2})\geq 2w\), and all \(v_{1},v_{2},u_{1},u_{2}\) lie in the same connected
component of \(G\setminus F\), contradicting that \(F\) is a \(2w\)-diameter decomposition.
Hence \(w \cdot p \geq \tfrac{8}{7}-\tfrac{4\epsilon}{7}\). \\

Suppose \(F\in A\).
Since \(d(r,u_{1})=d(r,u_{2})=k=\tfrac{w}{2}\), we must have
\(F\cap E(P_{u_{1}})\neq\varnothing\) and \(F\cap E(P_{u_{2}})\neq\varnothing\). Let \(P'_{a},P'_{b}\) be the two \(r\)–\(\tilde{r}\) paths on the cycle containing \(P_{a},P_{b}\), respectively.
Define \(A_{a}=\{F\in A \mid F \cap E(P_{a})\neq \varnothing\}\).
If \(F\in A_{a}\), then \(|F\cap E(P'_{a})|\geq 2\).
Since \(P'_{a}\) is a shortest path of length \(w\) from \(r\), by Lemma~\ref{lemma:upperbound},
\begin{align}\label{Aa in 16/7}
    \sum_{\substack{F \in \mathcal{F}\\ |F\cap E(P'_{a})| \geq 2}} y_{F} \;\leq\; (w\cdot p)-1
    \quad\Longrightarrow\quad
    \sum_{F\in A_{a}} y_{F} \;\leq\; (w\cdot p)-1.
\end{align}
Similarly, with \(A_{b}=\{F\in A \mid F \cap E(P_{b})\neq \varnothing \}\) and the path \(P'_{b}\),
\begin{align}\label{Ab in 16/7}
    \sum_{F\in A_{b}} y_{F} \;\leq\; (w\cdot p)-1.
\end{align}

\noindent Next, set
\[
A_{c}=\{F\in A \mid F\cap E(H'_{1})\in \mathcal{F}_{1}^{k}\},
\qquad
A_{d}=\{F\in A \mid F\cap E(H'_{2})\in \mathcal{F}_{2}^{k}\}.
\]

\begin{Claim}

\begin{align}\label{Ac Ad 16/7}
    \sum_{F\in A_{c}} y_{F}
< R^{k}_{(\mathcal{G},2w,p)} + \epsilon - \Bigl(1-\tfrac{w\cdot p}{2}\Bigr),
\quad
\sum_{F\in A_{d}} y_{F}
< R^{k}_{(\mathcal{G},2w,p)} + \epsilon - \Bigl(1-\tfrac{w\cdot p}{2}\Bigr).
\end{align}

\end{Claim}

\begin{proof}[Proof of Claim]
\renewcommand{\qedsymbol}{$\blacktriangle$}

\noindent We prove the inequality for \(A_{c}\); the case of \(A_{d}\) follows by the same reasoning.  
Since \(y(H'_{1})\) is an \(SDD(H'_{1},2w,p,r'_{1})\) and \(H'_{1}\) is a copy of \(H'\),
the choice of \((H',r')\) in \eqref{choice of H'} implies
\[
\sum_{F'\in \mathcal{F}^{k}_{1}} y(H'_{1})_{F'} 
\;<\; R^{k}_{(\mathcal{G},2w,p)} + \epsilon.
\]
Therefore,
\begin{align*}
\sum_{\substack{F\in \mathcal{F}\\ F\cap E(H'_{1})\in \mathcal{F}_{1}^{k}}} y_{F}
&=\sum_{F'\in \mathcal{F}^{k}_{1}} y(H'_{1})_{F'} \\
&=\sum_{F'\in \mathcal{F}^{k}_{1}}
       \sum_{\substack{F\in \mathcal{F}\\ F\cap E(H'_{1})=F'}} y_{F}
\;<\; R^{k}_{(\mathcal{G},2w,p)} + \epsilon.
\end{align*}

\noindent Note that \(A_{c}\subseteq \{F\in \mathcal{F}\mid F\cap E(H'_{1})\in \mathcal{F}_{1}^{k}\}\).
Consider instead
\[
A'_{a}=\{F\in \mathcal{F}\mid F\cap E(P_{u_{1}})=\varnothing\}.
\]
Then \(A'_{a}\subseteq \{F\in \mathcal{F}\mid F\cap E(H'_{1})\in \mathcal{F}_{1}^{k}\}\) and
\(A_{c}\cap A'_{a}=\varnothing\).
Moreover, since 
\(\mathcal{F}=A'_{a}\,\dot\cup\,\{F\in \mathcal{F}\mid F\cap E(P_{u_{1}})\neq\varnothing\}\),
we obtain, by applying the edge bound along the path \(P_{u_{1}}\),
\[
\sum_{\{F\in \mathcal{F}\mid F\cap E(P_{u_{1}})\neq\varnothing\}} y_{F}
\;\leq\;
\sum_{e\in E(P_{u_{1}})} \sum_{\substack{F\in \mathcal{F}\\ e\in F}} y_{F}
\;\leq\; \tfrac{w}{2}\,p \;=\; \tfrac{w\cdot p}{2}.
\]
Hence, since 
\[
1=\sum_{F\in \mathcal{F}}y_{F}
  =\sum_{F\in A'_{a}}y_{F}
  +\sum_{\{F\in \mathcal{F}\mid F\cap E(P_{u_{1}})\neq\varnothing\}}y_{F},
\]
we conclude
\[
\sum_{F\in A'_{a}} y_{F} \;\geq\; 1-\tfrac{w\cdot p}{2}.
\]

\noindent Consequently,
\begin{align*}
\sum_{F\in A_{c}} y_{F}
&\;\leq\;
\sum_{\substack{F\in \mathcal{F}\\ F\cap E(H'_{1})\in \mathcal{F}_{1}^{k}}} y_{F}
\;-\;
\sum_{F\in A'_{a}} y_{F} \\
&\;<\;
R^{k}_{(\mathcal{G},2w,p)} + \epsilon - \Bigl(1-\tfrac{w\cdot p}{2}\Bigr).
\nonumber
\end{align*}

\end{proof}

\noindent
Combining \eqref{Aa in 16/7}, \eqref{Ab in 16/7}, \eqref{Ac Ad 16/7}, we conclude

\begin{equation}\label{all Aa Ab Ac Ad together}
    \begin{aligned}
\sum_{F\in A_{a}\cup A_{b}\cup A_{c}\cup A_{d}} y_{F}\leq  2\,(w\cdot p - 1) + 2\Bigl(R^{k}_{(\mathcal{G},2w,p)} + \epsilon - (1-\tfrac{w\cdot p}{2})\Bigr) \\
= 3w\cdot p + 2R^{k}_{(\mathcal{G},2w,p)} + 2\epsilon - 4. 
\end{aligned}
\end{equation}

\noindent Let \(A^{*}=A\setminus(A_{a}\cup A_{b}\cup A_{c}\cup A_{d})\).
Using \eqref{A 16/7} and \eqref{all Aa Ab Ac Ad together}, we get
\begin{align*}
\sum_{F\in A^{*}} y_{F}
= \sum_{F\in A} y_{F}
   \;-\;
   \sum_{F\in A_{a}\cup A_{b}\cup A_{c}\cup A_{d}} y_{F} \geq
   R^{k}_{(\mathcal{G},2w,p)}
   \;-\;
   \Bigl[3w\cdot p + 2R^{k}_{(\mathcal{G},2w,p)} + 2\epsilon - 4\Bigr] \\
= (4 - 2\epsilon) - 3\,w\cdot p - R^{k}_{(\mathcal{G},2w,p)}.
\end{align*}
By the upper bound in Definition~\ref{def:R-G-2w-p-k}, we have
\(R^{k}_{(\mathcal{G},2w,p)} \le \tfrac{w\cdot p}{2}\).  
Therefore,
\begin{align*}
\sum_{F\in A^{*}} y_{F}
&\geq (4 - 2\epsilon) - 3\,w\cdot p - \tfrac{w\cdot p}{2}
 = (4 - 2\epsilon) - \tfrac{7}{2}\,w\cdot p >\; 0,
\end{align*}

\noindent where the last inequality uses the assumption \(w\cdot p < \tfrac{8}{7}-\tfrac{4\epsilon}{7}\).
Thus \(A^{*}\neq\varnothing\), so there exists \(F\in A^{*}\subseteq A\) with \(y_{F}>0\) such that
\(F\cap E(P_{a})=F\cap E(P_{b})=\varnothing\) and
\(F\cap E(H'_{i})\notin \mathcal{F}_{i}^{k}\) for \(i=1,2\), completing the proof.

\end{proof}

\begin{Corollary}\label{Corollary 16/7}
Let \(\mathcal{G}\) be the family of cactus graphs. 
If \(\mathcal{G}\) is strongly \(SDD(2w,p)\)-acceptable,
then \(w \cdot p \geq \tfrac{8}{7}\).    
\end{Corollary}

\begin{proof}
Recall that before Theorem~\ref{Theorem Cactus 16/7} we fixed a constant \(\epsilon>0\), 
and in that theorem we proved that
\[
w \cdot p \;\geq\; \tfrac{8}{7}-\tfrac{4\epsilon}{7}.
\]
Since this bound holds for every \(\epsilon>0\), taking the limit as \(\epsilon \to 0^{+}\) yields
\[
w \cdot p \;\geq\; \tfrac{8}{7},
\]
as claimed.
\end{proof}

\section{An Explicit Construction for the $\frac{20}{9}$ Lower Bound}\label{sec explicit example 20/9}
Inspired by the proof presented in \cref{sec:proof_main_thm}, we now give an explicit example on a cactus graph where the integrality gap is at least $\frac{20}{9}$. In particular, given an $\epsilon>0$, we are going to give an instance of the minimum multicut problem $M$ on a cactus graph $G$ such that, $$\frac{\text{OPT}_{IP}(M)}{\text{OPT}_{LP}(M)}\geq \frac{20}{9}-\epsilon.$$ 
We will construct the graph $G$ by using two $1-sum$ operations. Let $H$ be the following graph,

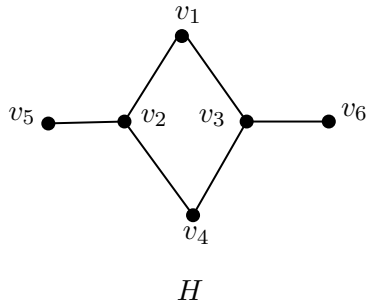
\begin{figure}[H]
    \centering
    \tikzset{every picture/.style={line width=0.75pt}} 

\begin{tikzpicture}[x=0.75pt,y=0.75pt,yscale=-1,xscale=1]

\draw    (242,86) -- (274.33,132) ;
\draw    (242,86) -- (213.33,132) ;
\draw    (247.57,179.13) -- (213.33,132) ;
\draw    (247.57,179.13) -- (274.33,132) ;
\draw    (213.33,132) -- (172.33,133) ;
\draw    (315.33,132) -- (274.33,132) ;
\draw  [fill={rgb, 255:red, 0; green, 0; blue, 0 }  ,fill opacity=1 ] (239,89) .. controls (239,87.34) and (240.34,86) .. (242,86) .. controls (243.66,86) and (245,87.34) .. (245,89) .. controls (245,90.66) and (243.66,92) .. (242,92) .. controls (240.34,92) and (239,90.66) .. (239,89) -- cycle ;
\draw  [fill={rgb, 255:red, 0; green, 0; blue, 0 }  ,fill opacity=1 ] (210.33,132) .. controls (210.33,130.34) and (211.68,129) .. (213.33,129) .. controls (214.99,129) and (216.33,130.34) .. (216.33,132) .. controls (216.33,133.66) and (214.99,135) .. (213.33,135) .. controls (211.68,135) and (210.33,133.66) .. (210.33,132) -- cycle ;
\draw  [fill={rgb, 255:red, 0; green, 0; blue, 0 }  ,fill opacity=1 ] (172.33,133) .. controls (172.33,131.34) and (173.68,130) .. (175.33,130) .. controls (176.99,130) and (178.33,131.34) .. (178.33,133) .. controls (178.33,134.66) and (176.99,136) .. (175.33,136) .. controls (173.68,136) and (172.33,134.66) .. (172.33,133) -- cycle ;
\draw  [fill={rgb, 255:red, 0; green, 0; blue, 0 }  ,fill opacity=1 ] (271.33,132) .. controls (271.33,130.34) and (272.68,129) .. (274.33,129) .. controls (275.99,129) and (277.33,130.34) .. (277.33,132) .. controls (277.33,133.66) and (275.99,135) .. (274.33,135) .. controls (272.68,135) and (271.33,133.66) .. (271.33,132) -- cycle ;
\draw  [fill={rgb, 255:red, 0; green, 0; blue, 0 }  ,fill opacity=1 ] (244.57,179.13) .. controls (244.57,177.48) and (245.91,176.13) .. (247.57,176.13) .. controls (249.23,176.13) and (250.57,177.48) .. (250.57,179.13) .. controls (250.57,180.79) and (249.23,182.13) .. (247.57,182.13) .. controls (245.91,182.13) and (244.57,180.79) .. (244.57,179.13) -- cycle ;
\draw  [fill={rgb, 255:red, 0; green, 0; blue, 0 }  ,fill opacity=1 ] (312.33,132) .. controls (312.33,130.34) and (313.68,129) .. (315.33,129) .. controls (316.99,129) and (318.33,130.34) .. (318.33,132) .. controls (318.33,133.66) and (316.99,135) .. (315.33,135) .. controls (313.68,135) and (312.33,133.66) .. (312.33,132) -- cycle ;

\draw (237,72.4) node [anchor=north west][inner sep=0.75pt]    {$v_{1}$};
\draw (220,124.4) node [anchor=north west][inner sep=0.75pt]    {$v_{2}$};
\draw (249,124.4) node [anchor=north west][inner sep=0.75pt]    {$v_{3}$};
\draw (154,123.4) node [anchor=north west][inner sep=0.75pt]    {$v_{5}$};
\draw (241,183.4) node [anchor=north west][inner sep=0.75pt]    {$v_{4}$};
\draw (320,121.4) node [anchor=north west][inner sep=0.75pt]    {$v_{6}$};
\draw (238,210.4) node [anchor=north west][inner sep=0.75pt]    {$H$};

\end{tikzpicture}
    \label{fig:enter-label}
\end{figure}

and $k$ be a sufficiently large natural number. Let $H_{1},\ldots,H_{k}$ be $k$ disjoint copies of $H$. Let $v_1^i$ be the vertex corresponding to $v_1$ in $H$. We construct $H'$ by taking 1-sum of $H_1,H_2,\ldots,H_k$ at $v_1^1,v_1^2,\ldots,v_1^k$, respectively. Let $v_1=v_1^1=v_1^2=\ldots=v_1^k$. We obtain $H''$ by adding an edge $(v_1,v_0)$ to $H'$. See the figure below for an illustration.

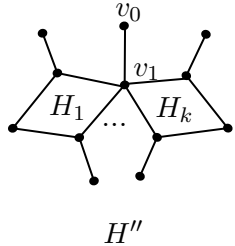
\begin{figure}[H]
    \centering
    \tikzset{every picture/.style={line width=0.75pt}} 

\begin{tikzpicture}[x=0.75pt,y=0.75pt,yscale=-1,xscale=1]

\draw    (224.27,162.46) -- (246.54,134.67) ;
\draw    (224.27,162.46) -- (257.48,167.14) ;
\draw    (246.54,134.67) -- (238.59,113.05) ;
\draw    (264.83,188.96) -- (257.48,167.14) ;
\draw  [fill={rgb, 255:red, 0; green, 0; blue, 0 }  ,fill opacity=1 ] (246.01,133.08) .. controls (247,132.74) and (248.05,133.18) .. (248.35,134.06) .. controls (248.65,134.95) and (248.08,135.93) .. (247.08,136.27) .. controls (246.08,136.61) and (245.03,136.16) .. (244.74,135.28) .. controls (244.44,134.4) and (245.01,133.41) .. (246.01,133.08) -- cycle ;
\draw  [fill={rgb, 255:red, 0; green, 0; blue, 0 }  ,fill opacity=1 ] (238.59,113.05) .. controls (239.59,112.72) and (240.64,113.16) .. (240.94,114.04) .. controls (241.23,114.92) and (240.67,115.91) .. (239.67,116.25) .. controls (238.67,116.58) and (237.62,116.14) .. (237.32,115.26) .. controls (237.02,114.38) and (237.59,113.39) .. (238.59,113.05) -- cycle ;
\draw  [fill={rgb, 255:red, 0; green, 0; blue, 0 }  ,fill opacity=1 ] (256.94,165.54) .. controls (257.94,165.21) and (258.99,165.65) .. (259.28,166.53) .. controls (259.58,167.41) and (259.01,168.4) .. (258.01,168.74) .. controls (257.02,169.07) and (255.97,168.63) .. (255.67,167.75) .. controls (255.37,166.87) and (255.94,165.88) .. (256.94,165.54) -- cycle ;
\draw  [fill={rgb, 255:red, 0; green, 0; blue, 0 }  ,fill opacity=1 ] (223.74,160.86) .. controls (224.74,160.53) and (225.79,160.97) .. (226.08,161.85) .. controls (226.38,162.73) and (225.81,163.72) .. (224.81,164.06) .. controls (223.81,164.39) and (222.76,163.95) .. (222.47,163.07) .. controls (222.17,162.19) and (222.74,161.2) .. (223.74,160.86) -- cycle ;
\draw  [fill={rgb, 255:red, 0; green, 0; blue, 0 }  ,fill opacity=1 ] (264.29,187.37) .. controls (265.29,187.03) and (266.34,187.47) .. (266.63,188.36) .. controls (266.93,189.24) and (266.36,190.22) .. (265.36,190.56) .. controls (264.36,190.9) and (263.31,190.46) .. (263.02,189.57) .. controls (262.72,188.69) and (263.29,187.7) .. (264.29,187.37) -- cycle ;
\draw    (331.62,162.38) -- (296.32,167.01) ;
\draw    (331.62,162.38) -- (310.9,136.01) ;
\draw    (296.32,167.01) -- (287.09,188.12) ;
\draw    (320.71,115.18) -- (310.9,136.01) ;
\draw  [fill={rgb, 255:red, 0; green, 0; blue, 0 }  ,fill opacity=1 ] (295.6,168.54) .. controls (294.64,168.09) and (294.19,167.04) .. (294.59,166.2) .. controls (294.99,165.36) and (296.08,165.04) .. (297.03,165.49) .. controls (297.99,165.94) and (298.44,166.98) .. (298.04,167.82) .. controls (297.65,168.67) and (296.55,168.99) .. (295.6,168.54) -- cycle ;
\draw  [fill={rgb, 255:red, 0; green, 0; blue, 0 }  ,fill opacity=1 ] (287.09,188.12) .. controls (286.13,187.67) and (285.68,186.62) .. (286.08,185.78) .. controls (286.47,184.94) and (287.57,184.62) .. (288.52,185.07) .. controls (289.47,185.52) and (289.93,186.56) .. (289.53,187.41) .. controls (289.13,188.25) and (288.04,188.57) .. (287.09,188.12) -- cycle ;
\draw  [fill={rgb, 255:red, 0; green, 0; blue, 0 }  ,fill opacity=1 ] (310.19,137.54) .. controls (309.23,137.09) and (308.78,136.04) .. (309.18,135.2) .. controls (309.57,134.36) and (310.67,134.04) .. (311.62,134.49) .. controls (312.57,134.94) and (313.02,135.98) .. (312.63,136.82) .. controls (312.23,137.67) and (311.14,137.99) .. (310.19,137.54) -- cycle ;
\draw  [fill={rgb, 255:red, 0; green, 0; blue, 0 }  ,fill opacity=1 ] (330.9,163.9) .. controls (329.95,163.45) and (329.5,162.41) .. (329.9,161.56) .. controls (330.29,160.72) and (331.39,160.4) .. (332.34,160.85) .. controls (333.29,161.3) and (333.74,162.35) .. (333.35,163.19) .. controls (332.95,164.03) and (331.86,164.35) .. (330.9,163.9) -- cycle ;
\draw  [fill={rgb, 255:red, 0; green, 0; blue, 0 }  ,fill opacity=1 ] (319.99,116.7) .. controls (319.04,116.25) and (318.58,115.21) .. (318.98,114.36) .. controls (319.38,113.52) and (320.47,113.2) .. (321.42,113.65) .. controls (322.38,114.1) and (322.83,115.15) .. (322.43,115.99) .. controls (322.04,116.83) and (320.94,117.15) .. (319.99,116.7) -- cycle ;
\draw  [fill={rgb, 255:red, 0; green, 0; blue, 0 }  ,fill opacity=1 ] (279.59,139.05) .. controls (280.59,138.72) and (281.64,139.16) .. (281.94,140.04) .. controls (282.23,140.92) and (281.67,141.91) .. (280.67,142.25) .. controls (279.67,142.58) and (278.62,142.14) .. (278.32,141.26) .. controls (278.02,140.38) and (278.59,139.39) .. (279.59,139.05) -- cycle ;
\draw    (296.32,167.01) -- (280.13,140.65) ;
\draw    (281.94,140.04) -- (310.19,137.54) ;
\draw    (280.13,140.65) -- (257.48,167.14) ;
\draw    (278.32,141.26) -- (246.54,134.67) ;
\draw    (280.33,110) -- (280.13,140.65) ;
\draw  [fill={rgb, 255:red, 0; green, 0; blue, 0 }  ,fill opacity=1 ] (279.29,109.37) .. controls (280.29,109.03) and (281.34,109.47) .. (281.63,110.36) .. controls (281.93,111.24) and (281.36,112.22) .. (280.36,112.56) .. controls (279.36,112.9) and (278.31,112.46) .. (278.02,111.57) .. controls (277.72,110.69) and (278.29,109.7) .. (279.29,109.37) -- cycle ;

\draw (268,207.4) node [anchor=north west][inner sep=0.75pt]    {$H''$};
\draw (283,128.4) node [anchor=north west][inner sep=0.75pt]    {$v_{1}$};
\draw (267,158.4) node [anchor=north west][inner sep=0.75pt]  [font=\large]  {$...$};
\draw (241,144.4) node [anchor=north west][inner sep=0.75pt]    {$H_{1}$};
\draw (294,145.4) node [anchor=north west][inner sep=0.75pt]    {$H_{k}$};
\draw (274,98.4) node [anchor=north west][inner sep=0.75pt]    {$v_{0}$};

\end{tikzpicture}
    \label{fig:enter-label}
\end{figure}

Let $H''_{1},\ldots,H''_{k}$ be $k$ disjoint copies of $H''$. Let $v_0^i$ be the vertex of $H^{''}_i$ corresponding to the vertex $v_0$ of $H''$. Let $G=(V,E)$ be the graph obtained by taking $1-sum$ of $H''_{1},\ldots,H''_{k}$ be at $v_0^1,\ldots,v_0^k$, respectively. Let $v_0=v_0^1=\ldots=v_0^k$. Let $v_{i}$ be the unique neighbor of $v_{0}$ in $H_{i}^{''}$. See the figure below for an illustration.

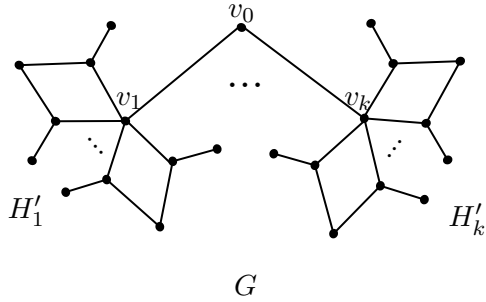
\begin{figure}[H]
    \centering
    \input{G_in_the_example_Figure}
    \label{fig:enter-label}
\end{figure}
Let $l(e)=1$ for all $e \in E$. We partition the set of edges $E$ based on their distance from $v_0$ as follows: $$E_{1}=\{e\in E~|~l(v_{0},e)=0\},~E_{2}=\{e\in E~|~l(v_{0},e)=1\},~E_{3}=\{e\in E~|~l(v_{0},e)=2\}.$$
Note that $E_{1}$ is the set of incident edges to $v_{0}$, and $E_{1},E_{2},E_{3}$ is a partition of $E$. Now we are going to define an instance of the minimum multicut problem on $G$. We first assign costs to edges as follows: 
\[
c(e) = \begin{cases}
            k & e\in E_{1}\\
      2 & e\in E_{2}\\
      1 & e\in E_{3}
         \end{cases}
\]
be the cost function, and let $S=\{(u,v) \in V \times V~|~l(u,v) \geq 4 \}$ be the set of source-sink pairs. We will denote this multicut instance by $M$. It is easy to see that $x=\{x_{e}=\frac{1}{4}~|~e\in E\}$ is a feasible fractional solution to $M$ with cost $$\dfrac{k \cdot k+ 2 \cdot k^{2}\cdot 2+ 4k^{2} \cdot 1}{4}=\dfrac{9 \cdot k^{2}}{4}.$$ We will now show that the cost of any feasible multicut (i.e.~an integral solution) is at least $5k^{2}-9k$. This will imply that the integrality gap for this instance is at least $\frac{20}{9}-\frac{4}{k}$, which can be arbitrarily close to $\frac{20}{9}$. More precisely, for any $\epsilon >0$, we can set $k>\frac{4}{\epsilon}$ to obtain a lower bound of $\frac{20}{9}-\epsilon$.

Recall that for a graph $H$, we use $V(H)$ and $E(H)$ to denote the set of vertices and edges in $H$, and for $E' \subseteq E(H) $, we use $c(E')$ to denote the total cost of edges in $E'$. Let $F$ be a feasible multicut solution to $M$. Let $t=|E_{1}\cap F|$. For now assume that $0<t<k$. Without loss of generality, assume that $(v_{0},v_{i})\in F$ for $i=1,\ldots,t$. 
\begin{Claim}
    $c(F \cap E(H''_{i})) \geq 4 \cdot (k-1)+k$ for $i=1,2,\ldots,t$.
\end{Claim}\label{claim:example}
\begin{proof}
    We will prove the claim for $i=1$. The proof for other values of $i$ is identical. Denote $H_{1},\ldots,H_{k}$ as the $k$ copies of $H$ incident at $v_{1}$. For each $1\leq i\leq k$, we call $H_{i}$ good if $rad_{F}(v_{1}) \leq 1$, and we call it bad otherwise. It is not too difficult to see that there is at most one bad $H_{i}$. Suppose that $H_{1},H_{2}$ are bad graphs, then there exists $a\in V(H_{1}), b\in V(H_{2})$ such that $l(a,v_{1}),l(b,v_{1})\geq 2$. But this is a contradiction since $a$ and $b$ are within the same component as $v_{1}$ after the removal of $F$, and are at a distance 4 apart, i.e.~they are a source-sink pair in the multicut instance $M$. Thus, there are at least $k-1$ good graphs attached to $v_{1}$. By doing a simple case analysis, it can be verified that $c(F \cap E(H_i))\geq 4$ if $H_i$ is good. Combining the above with the fact that the edge between $v_{0},v_{1}$ is included in $F$, we obtain the statement of the claim.
\end{proof}
Note that $v_{t+1},\ldots,v_{k}$ are within the same component as $v_{0}$ after the removal of $F$. For each $t+1\leq j\leq k$, we call $H''_{j}$ \emph{good} if $rad_{F}(v_{0})\leq 1$ in $H''_{j}$, otherwise we call it \emph{bad}. Using a similar argument as in the proof of Claim~\ref{claim:example}, one can show that there at most $1$ bad $H''_{j}$. Thus, we have at least $k-t-1$ good $H''_{j}$'s. 

\begin{Claim}
    $c(F \cap E(H''_{j})) \geq 5 \cdot k$ if $H''_{j}$ is \emph{good} for $t+1\leq j\leq k$.
\end{Claim}
\begin{proof}
    Since the edge between $v_{0},v_{j}$ is not included in $F$, $E_{2}\cap E(H''_{j})\subseteq F$ or equivalently, all the edges with cost $2$ of $H''_{j}$ are included in $F$. On the other hand, let $H_{i}$ be one of the attached copies of $H$ to $v_{j}$. Note that we have already showed that $E_{2}\cap E(H_i)\in F$. If $E_{3}\cap E(H_i) \cap F = \varnothing$, then there is a source-sink pair at distance 4 in $H_i$ which is not disconnected. 
     Therefore, in each $H_i$, $F$ picks edges of total weight at least $5$. 
\end{proof}
\noindent Therefore we obtain, $$c(F)\geq t \cdot (5k-4)+(k-1-t) \cdot (5k)=5k^{2}-5k-4t\geq 5k^{2}-9k.$$ Even when $t=0,k$, one can use the same arguments as above to obtain the same lower bound on the cost of $F$. This concludes the proof of the theorem.
\section{Conclusions and Future Work}
We improve upon a decade-old lower bound on the multiflow-multicut gap for planar graphs and, in doing so, develop new techniques. The main question our work raises is whether tight gap results can be obtained, even for the class of cactus and series-parallel graphs, and more generally, for planar graphs. Proving such a result likely requires new techniques, making it an interesting and challenging problem.\\

\noindent \textbf{Acknowledgment}: We would like to thank Joseph Cheriyan for many helpful discussions throughout the course of this project.

\printbibliography

@InProceedings{kalantarzadeh_et_al:LIPIcs.APPROX/RANDOM.2025.14,
  author =	{Kalantarzadeh, Sina and Kumar, Nikhil},
  title =	{{Improved Lower Bounds on Multiflow-Multicut Gaps}},
  booktitle =	{Approximation, Randomization, and Combinatorial Optimization. Algorithms and Techniques (APPROX/RANDOM 2025)},
  pages =	{14:1--14:18},
  series =	{Leibniz International Proceedings in Informatics (LIPIcs)},
  ISBN =	{978-3-95977-397-3},
  ISSN =	{1868-8969},
  year =	{2025},
  volume =	{353},
  editor =	{Ene, Alina and Chattopadhyay, Eshan},
  publisher =	{Schloss Dagstuhl -- Leibniz-Zentrum f{\"u}r Informatik},
  address =	{Dagstuhl, Germany},
  URL =		{https://drops.dagstuhl.de/entities/document/10.4230/LIPIcs.APPROX/RANDOM.2025.14},
  URN =		{urn:nbn:de:0030-drops-243803},
  doi =		{10.4230/LIPIcs.APPROX/RANDOM.2025.14}
}

@book{SchrijverFarkasLemma,
    author = {Schrijver},
    title = {Theory of Linear and Integer Programming},
    publisher = {Wiley},
    year = {1998},
    page = {90}
}

@article{VempalaConvexCombination,
    author = {Carr, Robert and Vempala, Santosh},
    title = {Randomized metarounding},
    journal = {Probabilistic Methods in Combinatorial Optimization},
    volume = {20},
    issue = {3},
    pages = {342-352},
    publisher = {Wiley},
    year = {2002}
}

@article{garg1996approximate,
  title={Approximate max-flow min-(multi) cut theorems and their applications},
  author={Garg, Naveen and Vazirani, Vijay V and Yannakakis, Mihalis},
  journal={SIAM Journal on Computing},
  volume={25},
  number={2},
  pages={235--251},
  year={1996},
  publisher={SIAM}
}

@article{garg1997primal,
  title={Primal-dual approximation algorithms for integral flow and multicut in trees},
  author={Garg, Naveen and Vazirani, Vijay V. and Yannakakis, Mihalis},
  journal={Algorithmica},
  volume={18},
  number={1},
  pages={3--20},
  year={1997},
  publisher={Springer}
}

@inproceedings{klein1993excluded,
  title={Excluded minors, network decomposition, and multicommodity flow},
  author={Klein, Philip and Plotkin, Serge A and Rao, Satish},
  booktitle={Proceedings of the twenty-fifth annual ACM symposium on Theory of computing},
  pages={682--690},
  year={1993},
  organization={ACM}
}

@article{tardos1993improved,
  title={Improved bounds for the max-flow min-multicut ratio for planar and {$K_{r,r}$}-free graphs},
  author={Tardos, Eva and Vazirani, Vijay V},
  journal={Information Processing Letters},
  volume={47},
  number={2},
  pages={77--80},
  year={1993},
  publisher={Elsevier}
}

@article{abraham2014cops,
  title={Cops, robbers, and threatening skeletons: Padded decomposition for minor-free graphs},
  author={Abraham, Ittai and Gavoille, Cyril and Gupta, Anupam and Neiman, Ofer and Talwar, Kunal},
  journal={SIAM Journal on Computing},
  volume={48},
  number={3},
  pages={1120--1145},
  year={2019},
  publisher={SIAM Journal on Computing}
}

@incollection{fakcharoenphol2003improved,
  title={An improved decomposition theorem for graphs excluding a fixed minor},
  author={Fakcharoenphol, Jittat and Talwar, Kunal},
  booktitle={Approximation, Randomization, and Combinatorial Optimization.. Algorithms and Techniques},
  pages={36--46},
  year={2003},
  publisher={Springer}
}

@inproceedings{rao1999small,
  title={Small distortion and volume preserving embeddings for planar and Euclidean metrics},
  author={Rao, Satish},
  booktitle={Proceedings of the Fifteenth Annual Symposium on Computational Geometry},
  pages={300--306},
  year={1999},
  organization={ACM}
}

@article{chekuri2013flow,
  title={Flow-cut gaps for integer and fractional multiflows},
  author={Chekuri, Chandra and Shepherd, F Bruce and Weibel, Christophe},
  journal={Journal of Combinatorial Theory, Series B},
  volume={103},
  number={2},
  pages={248--273},
  year={2013},
  publisher={Elsevier}
}

@article{gupta2004cuts,
  title={Cuts, trees and {$l_1$}-embeddings of graphs},
  author={Gupta, Anupam and Newman, Ilan and Rabinovich, Yuri and Sinclair, Alistair},
  journal={Combinatorica},
  volume={24},
  number={2},
  pages={233--269},
  year={2004},
  publisher={Springer}
}

@inproceedings{chakrabarti2008embeddings,
  title={Embeddings of topological graphs: lossy invariants, linearization, and 2-sums},
  author={Chakrabarti, Amit and Jaffe, Alexander and Lee, James R and Vincent, Justin},
  booktitle={49th Annual IEEE Symposium on Foundations of Computer Science, 2008},
  pages={761--770},
  year={2008},
  organization={IEEE}
}

@article{lee2010coarse,
  title={Coarse differentiation and multi-flows in planar graphs},
  author={Lee, James R and Raghavendra, Prasad},
  journal={Discrete \& Computational Geometry},
  volume={43},
  number={2},
  pages={346--362},
  year={2010},
  publisher={Springer}
}

@Article{Linial1995,
author="Linial, Nathan
and London, Eran
and Rabinovich, Yuri",
title="The geometry of graphs and some of its algorithmic applications",
journal="Combinatorica",
year="1995",
month="Jun",
day="01",
volume="15",
number="2",
pages="215--245",

}

@article{hu1963multi,
  title={Multi-commodity network flows},
  author={Hu, T Chiang},
  journal={Operations research},
  volume={11},
  number={3},
  pages={344--360},
  year={1963},
  publisher={INFORMS}
}

@article{filtser2024optimal,
  title={Optimal padded decomposition for bounded treewidth graphs},
  author={Filtser, Arnold and Friedrich, Tobias and Issac, Davis and Kumar, Nikhil and Le, Hung and Mallek, Nadym and Zeif, Ziena},
  journal={arXiv preprint arXiv:2407.12230},
  year={2024}
}

@inproceedings{friedrich2023approximate,
  title={Approximate max-flow min-multicut theorem for graphs of bounded treewidth},
  author={Friedrich, Tobias and Issac, Davis and Kumar, Nikhil and Mallek, Nadym and Zeif, Ziena},
  booktitle={Proceedings of the 55th Annual ACM Symposium on Theory of Computing},
  pages={1325--1334},
  year={2023}
}

@article{ford1956maximal,
  title={Maximal flow through a network},
  author={Ford, Lester Randolph and Fulkerson, Delbert R},
  journal={Canadian journal of Mathematics},
  volume={8},
  pages={399--404},
  year={1956},
  publisher={Cambridge University Press}
}

@inproceedings{conroy2025protect,
  title={How to protect yourself from threatening skeletons: Optimal padded decompositions for minor-free graphs},
  author={Conroy, Jonathan and Filtser, Arnold},
  booktitle={Proceedings of the 57th Annual ACM Symposium on Theory of Computing},
  pages={2281--2292},
  year={2025}
}

\end{document}